\documentclass[12pt]{article}
\usepackage{amsmath, amssymb, amsthm}

\def\w{\mbox{$\omega$}}
\def\x{{\mathbf x}}
\def\y{{\mathbf y}}

\def\0{{\mathbf 0}}
\def\1{{\mathbf 1}}
\def\r{{\mathbf r}}  

\def\a{{\mathbf a}}
\def\b{{\mathbf b}}
\def\c{{\mathbf c}}

\def\F{{\mathbb F}}

\def\Z{{\mathbb Z}}

\def\Fq{{\mathbb F}_q}
\def\z{{\zeta}}

\def\Zp8{{\Z_{p^\infty}}}
\def\R{{\mathcal R}}

\newtheorem{thm}{Theorem}[section]

\newtheorem{lem}[thm]{Lemma}
\newtheorem{cor}[thm]{Corollary}
\theoremstyle{definition}

\newtheorem{rem}[thm]{Remark}

\newenvironment{proof*}{\noindent {\bf Proof of Theorem 3.1 \ \ }}{\hfill $\square$\medskip}

\newcommand{\ob}{\mbox{$\overline{\omega}$}}
\newcommand{\om}{\mbox{$\omega$}}

\def\rank{\operatorname{rank}}

\date{December 17, 2011}

\begin{document}

\newcommand{\comment}[1]{} 

\title{Construction of quasi-cyclic self-dual codes}

\author{
Sunghyu Han\thanks{\scriptsize School of Liberal Arts, Korea
University of Technology and Education, Cheonan 330-708, South
Korea, Email: sunghyu@kut.ac.kr}, Jon-Lark Kim\thanks{\scriptsize
Department of Mathematics, University of Louisville, Louisville, KY
40292, USA, Email: jl.kim@louisville.edu}, Heisook
Lee\thanks{\scriptsize Department of Mathematics, Ewha Womans
University, Seoul 120-750, South Korea, Email: hsllee@ewha.ac.kr},
and Yoonjin Lee\thanks{\scriptsize Department of Mathematics, Ewha
Womans University, Seoul 120-750, South Korea, Email:
yoonjinl@ewha.ac.kr} {\thanks{\scriptsize The author is a
corresponding author and supported by the National Research Foundation of
Korea(NRF) grant funded by the Korea government(MEST) (No.
2010-0015201).}}} 

\maketitle

\begin{abstract}\noindent
There is a one-to-one correspondence between $\ell$-quasi-cyclic
codes over a finite field $\mathbb F_q$ and linear codes over a ring
$R = \mathbb F_q[Y]/(Y^m-1)$. Using this correspondence, we prove
that every $\ell$-quasi-cyclic self-dual code of length $m\ell$ over
a finite field $\mathbb F_q$ can be obtained by the {\it
building-up} construction, provided that char $(\F_q)=2$ or $q
\equiv 1 \pmod 4$, $m$ is a prime $p$, and $q$ is a primitive
element of $\F_p$. We determine possible weight enumerators of a
binary $\ell$-quasi-cyclic self-dual code of length $p\ell$ (with
$p$ a prime) in terms of divisibility by $p$. We improve the result
of~\cite{cubic} by constructing new binary cubic (i.e.,
$\ell$-quasi-cyclic codes of length $3\ell$) optimal self-dual codes of
lengths $30, 36, 42, 48$ (Type I), $54$ and $66$. We also find
quasi-cyclic optimal self-dual codes of lengths $40$, $50$, and
$60$. When $m=5$, we obtain a new $8$-quasi-cyclic self-dual $[40,
20, 12]$ code over $\F_3$ and a new $6$-quasi-cyclic self-dual $[30,
15, 10]$ code over $\F_4$. When $m=7$, we find a new
$4$-quasi-cyclic self-dual $[28, 14, 9]$ code over $\F_4$ and a new
$6$-quasi-cyclic self-dual $[42,21,12]$ code over $\F_4$.
\end{abstract}




\section*{Introduction}
Self-dual codes have been one of the most interesting classes of
linear codes over finite fields and in general over finite rings.
They interact with other areas including
lattices~\cite{ConSlo_sphere, Ebe}, invariant
theory~\cite{NebRaiSlo}, and designs~\cite{AssKey}. On the other
hand, quasi-cyclic codes have been one of the most practical classes
of linear codes. Linear codes which are quasi-cyclic and self-dual
simultaneously are an interesting class of codes, and this class of
codes is our main topic. We refer to~\cite{HufPle} for a basic
discussion of codes.

From the module theory over rings, quasi-cyclic codes can be considered as
modules over the group algebra of the cyclic group.
For a special ring $R = \mathbb F_q[Y]/(Y^m-1)$,
Ling and Sol\'{e} \cite{QCI, QCIII} consider linear codes over a ring $R$,
where $m$ is a positive integer coprime to $q$, and they use a
correspondence $\phi$ between (self-dual) quasi-cyclic codes over $\F_q$ and (self-dual, respectively) linear codes over $R$.
We call quasi-cyclic codes over $\F_q$ {\it{cubic}}, {\it{quintic}}, or {\it{septic}} codes depending
on $m=3, 5,$ or $7$, respectively. Bonnecaze et.~al.~\cite{cubic} studied binary cubic
self-dual codes, and Bracco et.~al.~\cite{BraNatSol} considered binary quintic self-dual codes.

In this paper, we focus on construction and classification of
quasi-cyclic self-dual codes over a finite field $\Fq$ under the
usual permutation or monomial equivalence. We note that the
equivalence under the correspondence $\phi$ may not be preserved;
two inequivalent linear codes over a ring $R$ under a permutation
equivalence may correspond to two equivalent quasi-cyclic codes over
a finite field $\Fq$ under a permutation or monomial equivalence.
Hence, we first construct all self-dual codes over the ring $R$
using a building-up construction. Rather than considering the
equivalence of these codes over $R$, we consider the equivalence of
their corresponding quasi-cyclic self-dual codes over $\Fq$ to get a
complete classification of quasi-cyclic self-dual codes over $\Fq$.

We prove that every $\ell$-quasi-cyclic self-dual code of length
$m\ell$ over $\mathbb F_q$ can be obtained by the {\it building-up
construction}, provided that char$(\F_q)=2$ or $q \equiv 1 \pmod 4$,
$m$ is a prime $p$, and $q$ is a primitive element of $\F_p$. Our
result shows that the building-up construction is a complete method
for constructing all $\ell$-quasi-cyclic self-dual codes of length
$m\ell$ over $\mathbb F_q$ subject to certain conditions of $m$ and
$q$. We determine possible weight enumerators of a binary
$\ell$-quasi-cyclic self-dual code of length $p\ell$ with $p$ a
prime in terms of divisibility by $p$.

By employing our building-up constructions, we classify binary cubic
self-dual codes of lengths up to $24$, and we construct
binary cubic optimal self-dual codes of lengths $30, 36, 42, 48$ (Type I),
$54$ and $66$. We point out that the advantage of our construction
is that we can classify all binary cubic self-dual codes in a more
efficient way without searching for all binary self-dual codes. We
summarize our result on the classification of binary cubic extremal
self-dual codes in Table~\ref{tab:bin_cubic_sd}. We also give a
complete classification of all binary quintic self-dual codes of
even lengths $5\ell \le 30$, and construct such optimal codes of
lengths $40$, $50$, and $60$. For various values of $m$ and $q$, we
obtain quintic self-dual codes of length $5 \ell$ over $\mathbb F_3$
and $\mathbb F_4$ and septic self-dual codes of length $7\ell$ over
$\F_2, \F_4$, and $\F_5$ which are optimal or have the best known
parameters. In particular, we find a new quintic self-dual $[40, 20,
12]$ code over $\F_3$ and a new quintic self-dual $[30, 15, 10]$
code over $\F_4$. We also obtain a new septic self-dual $[28, 14,
9]$ code over $\F_4$ and a new septic self-dual $[42,21,12]$ code
over $\F_4$.

\comment{
The main result of this paper is the following: Let char $(\F_q)=2$ or $q \equiv 1 \pmod 4$
and $R = \F_q[Y]/(Y^m-1)$, where $m$ is a prime $p$ and $q$ is a primitive element of $\F_p$.
Every self-dual code over the ring $R$ (equivalently, $\ell$-quasi-cyclic self-dual codes of
length $m\ell$ over $\mathbb F_q$) can be obtained by the building-up construction.
This result shows that the building-up construction for this case is a complete method for constructing all the self-dual codes over the ring $R$ (equivalently, $\ell$-quasi-cyclic self-dual codes of length $m\ell$ over $\mathbb F_q$).
Using this result, we find a complete classification of all binary {\it quintic}
(i.e., $\ell$-quasi-cyclic codes of length $5\ell$) self-dual codes of lengths up to $30$, and
we also construct optimal binary quintic self-dual codes of lengths $40$ (Type II), $50$, and $60$.
Furthermore, we find interesting quasi-cyclic self-dual codes by investigating various cases of $m$ and $q$, and they are quintic
self-dual codes over $\mathbb F_3$ and $\mathbb F_4$ and septic
(i.e., $\ell$-quasi-cyclic codes of length $7\ell$) self-dual codes over
$\F_2, \F_4$, and $\F_5$ which are optimal or have the best known parameters.
In fact, we obtain a new quintic self-dual $[40, 20, 12]$ code over $\F_3$ and
new septic self-dual codes over $\F_4$ with parameters
$[28, 14, 9]$ and $[42, 21, 12]$.
}

\begin{table}
\centering \caption{Binary extremal cubic self-dual codes of lengths
up to $66$} \label{tab:bin_cubic_sd}
\[
\begin{tabular}{c|c|c|l}
\hline

\hline length $n$ & highest min. wt & No. of extremal & Ref.\\
& & cubic self-dual codes \\ \hline
6  &  2 & 1 & Sec.~\ref{sec:construction}\\
12 &  4 & 1 & Sec.~\ref{sec:construction}\\
18 &  4 & 1 & Sec.~\ref{sec:construction}\\
24 &  8 & 1 & Sec.~\ref{sec:construction}\\
30 &    6 & $8$ & Sec.~\ref{sec:construction}, \cite{cubic}, \cite{Mun2010Web}\\
36 &    8 & $13$ & Sec.~\ref{sec:construction}, \cite{cubic}, \cite{Gab2011_Web}, \cite{HarMun2010} \\
42 &    8 & $1569$& Sec.~\ref{sec:construction}, \cite{cubic}, \cite{BHM}, \cite{BouYanRus} \\
48 &    10& $\ge 4$ & Sec.~\ref{sec:construction}, \cite{cubic}\\
54 &    10& $\ge 7$ & Sec.~\ref{sec:construction}, \cite{cubic}\\
60 &    12& $\ge 3$ & \cite{cubic}\\
66 &    12& $\ge 7$ & Sec.~\ref{sec:construction}, \cite{cubic}\\

\hline

\hline
\end{tabular}
\]
\end{table}

This paper is organized as follows. Section~\ref{sec:Preliminaries}
contains some basic notations and definitions, and
Section~\ref{sec:Building-up} presents the building-up construction
method of quasi-cyclic self-dual codes over finite fields. In
Section~\ref{sec:construction}, we construct binary quasi-cyclic
self-dual codes, and we find the cubic codes and
quintic codes. In
Section~\ref{sec:construction-various}, we construct quasi-cyclic
self-dual codes over various fields such as $\F_2$, $\F_3, \F_4$,
and $\mathbb F_5$, and we obtain the cubic codes, the quintic codes
and the septic codes. We use Magma~\cite{CanPla} for
computations.

\section{Preliminaries} \label{sec:Preliminaries}

We briefly introduce some basic notions about quasi-cyclic self-dual
codes. For more detailed description, we refer to~\cite{QCI, QCIII}.

Let $R$ be a commutative ring with identity. A {\it linear code} $C$
of length $n$ over $R$ is defined to be an $R$-submodule of $R^n$;
in particular, if $R$ is a finite field $\F_q$ of order $q$, then
$C$ is a vector subspace of $\F_q^n$ over $\F_q.$ The dual of
$C$ is denoted by $C^{\perp}$, $C$ is \textit{self-orthogonal} if $C
\subseteq C^{\perp}$, and \textit{self-dual} if $C = C^{\perp}$. We
denote the standard shift operator on $R^n$ by $T$. A linear code
$C$ is said to be \emph{quasi-cyclic of index $\ell$} or {\it
$\ell$-quasi-cyclic} if it is invariant under
$T^{\ell}$. A $1$-quasi-cyclic code means a cyclic code. Throughout
this paper, we assume that the index $\ell$ divides the code length
$n$.

Let $m$ be a positive integer coprime to the characteristic of
$\F_q$, $\F_q[Y]$ be a polynomial ring, and $R := R(\F_q, m) =
\F_q[Y]/(Y^m - 1)$. Then it is shown \cite{QCI} that there is a one-to-one
correspondence between $\ell$-quasi-cyclic codes over $\F_q$ of
length $\ell m$ and linear codes over $R$ of length $\ell$, and the
correspondence is given by the map $\phi$ defined as follows. Let
$C$ be a quasi-cyclic code over $\F_q$ of length $lm$ and index $l$
with a codeword $\c$ denoted by
$\c =(c_{00}, c_{01}, \dots, c_{0,\ell-1},
      c_{10}, \dots, c_{1,\ell-1}, \dots,
       c_{m-1,0}, \dots, c_{m-1,\ell-1}).$
Let $\phi$ be a map $\phi: {\F_q}^{\ell m} \rightarrow R^{\ell}$ defined by
\begin{equation*}
\phi(\c) = (\c_{0}(Y), \c_{1}(Y), \dots, \c_{\ell-1}(Y)) \in R^{\ell},
\end{equation*}
where
$\c_j(Y) = \sum_{i=0}^{m-1}c_{ij}Y^i \in R, {\mbox{ for }} j=0, \dots, \ell-1.$
We denote by $\phi(C)$ the image of $C$ under $\phi$.

A {\it conjugation} map $^{-}$ on $R$ is defined as the map that
sends $Y$ to $Y^{-1}=Y^{m-1}$ and acts as the identity map on $\F_q$, and
it is extended $\F_q$-linearly.
On $R^{\ell}$, we define the \textit{Hermitian inner product}
by \ $\langle \x, \y \rangle = \sum_{j=0}^{\ell-1}x_{j}\overline{y_{j}}$ \ \
for $\x = (x_0, \dots, x_{\ell -1})$ and $\y = (y_0, \dots, y_{\ell -1})$.

It is proved \cite{QCI} that for $\a, \b \in \F_q^{\ell m}$, \
$T^{\ell k}(\a)\cdot \b = 0$ \ for all $0 \leq k \leq m-1$ \ if and
only if \ $\langle \phi(\a), \phi(\b) \rangle = 0$, where $\cdot$
denotes the standard Euclidean inner product. From this fact, it
follows that $\phi(C)^{\perp} = \phi(C^{\perp})$, where
$\phi(C)^{\perp}$ is the dual of $\phi(C)$ with respect to the
Hermitian inner product, and $C^{\perp}$ is the dual of $C$ with
respect to the Euclidean inner product. In particular, a
quasi-cyclic code $C$ over $\F_q$ is self-dual with respect to the
Euclidean inner product if and only if $\phi(C)$ is self-dual over
$R$ with respect to the Hermitian inner product~\cite{QCI}. Two
linear codes $C_1$ and $C_2$ over $R$ are {\it equivalent} if there
is a permutation of coordinates of $C_1$ sending $C_1$ to $C_2$.
Similarly, two linear codes over $\F_q$ are equivalent if there is a
monomial mapping sending one to another. Note that the equivalence
of two linear codes $C_1$ and $C_2$ over $R$ implies a permutation
equivalence of quasi-cyclic linear codes $\phi^{-1}(C_1)$ and
$\phi^{-1}(C_2)$ over $\Fq$, but not conversely in general.


\comment{
It is known~\cite{QCI} that a binary linear code $C$ contains a fixed point free element of
order three if and only if $C$ is obtained by a generalized cubic construction from a binary
code and a quaternary code both of length $\ell$, and this is exactly the case $m=3$
in the above correspondence given by $\phi.$ Therefore, for a quasi-cyclic code $C$ over $F$
of length $lm$ and index $l$, as in~\cite{cubic}, when $m=3$, we call $C$ a {\it cubic code},
and when $m=5$, we call $C$ a {\it quintic code} throughout this paper.
} 

\section{Construction of quasi-cyclic self-dual codes}
\label{sec:Building-up}

Throughout this paper, let $R = \F_q[Y]/(Y^m -1)$, and self-dual (or self-orthogonal) codes
over $R$ means self-dual (or self-orthogonal) codes with respect to the Hermitian inner product.

We begin with the following lemma regarding the length of self-dual codes.
\begin{lem}
Let $R = \F_q[Y]/(Y^m - 1)$.
\begin{enumerate}
\item If $char(\F_q) = 2$ or $q \equiv 1 \pmod 4$, then
there exists a self-dual code over $R$ of length $\ell$
if and only if $2 \mid \ell$.
\item If $q \equiv 3 \pmod 4$, then
there exists a self-dual code over $R$ of length $\ell$
if and only if $4 \mid \ell$.
\end{enumerate}
\end{lem}

\begin{proof}

To prove (i) and (ii), we observe the following.
Suppose $C$ is a self-dual code of length $\ell$ over $R$. We may
assume that $C_1$ in the decomposition of $C$ in~\cite[Theorem
4.2]{QCI} is a Euclidean self-dual code over $\F_q$ of length
$\ell$.

For (i), suppose that $char(\F_q) = 2$ or $q \equiv 1 \pmod 4$.
By the above observation, $2 \mid \ell$. Conversely, let $\ell = 2k$. We take a
Euclidean self-dual code over $\F_q$ of length $2$ using the
following generator matrix: $[1 ~~c]$, where $c^2 = -1$. We can see
that this matrix generates a self-dual code $C$ over $R$ of length
$2$. Then the direct sum of the $k$ copies of $C$ is a self-dual code over $R$
of length $\ell = 2k$.

For (ii), let $q \equiv 3 \pmod 4$.
It is well known~\cite[p. 193]{RaiSlo} that if $q \equiv 3 \pmod 4$ then
a self-dual code of length $n$ exists if and only if $n$ is a
multiple of $4$. Hence by the above observation,  $4 \mid \ell$.
Conversely, let $\ell = 4k$ for some positive integer $k$. It is
known~\cite[p. 281]{IreRos} that if $q$ is a power of an odd prime
with $q \equiv 3 \pmod 4$, then there exist nonzero $\alpha$ and
$\beta$ in $\F_q$ such that $\alpha^2+\beta^2+1=0$ in $\F_q$. We
take a Euclidean self-dual code over $\F_q$ of length $4$ with the
following generator matrix:
\begin{equation*}
G= \left[
  \begin{array}{cccc}
    1 & 0 & \alpha & \beta  \\
    0 & 1 & -\beta & \alpha    \\
  \end{array}
\right],
\end{equation*}
where $\alpha^2+\beta^2+1=0$ in $\F_q$. We can see that this matrix
generates a self-dual code $C$ over $R$ of length $4$. Then the direct sum of the $k$ copies of
$C$ is a self-dual code over $R$ of length $\ell
= 4k$.
\end{proof}

The following theorem is the building-up constructions for self-dual codes over $R$,
equivalently, $\ell$-quasi-cyclic self-dual codes over $\F_q$ for any odd prime power $q$.
The proof is similar to that of~\cite{KL}, so the proof is omitted.

\begin{thm} \label{thm:construction1}
Let $C_0$ be a self-dual code over $R$ of length $2\ell$ and
$G_0 = (\r_i)$ be a $k \times 2\ell$ generator matrix for $C_0$, where
$\r_i$ is the $i$-th row of $G_0$, $1 \leq i \leq k$.
\begin{enumerate}

\item  Assume that char $(\F_q)=2$ or $q \equiv 1 \pmod 4$.\\
Let $c$ be in $R$ such that $c\overline{c} = -1$, $\x$ be a vector in $R^{2\ell}$
with $\langle \x, \x \rangle = -1$, and $y_i = -{\langle \r_i, \x \rangle}$ for $1 \leq i \leq k$.
Then the following matrix
\begin{equation*}
G=
\left[
  \begin{array}{cc|c}
    1 & 0 & ~~~~~~~~ \x ~~~~~~~~ \\
    \hline
    y_1 & cy_1 &   \r_1 \\
    \vdots & \vdots & \vdots   \\
    y_{k} & cy_{k} & \r_{k}   \\
  \end{array}
\right]
\end{equation*}
generates a self-dual code $C$ over $R$ of length $2\ell+2$.

\item Assume that $q \equiv 3 \pmod 4$ and $\ell$ is even.\\
Let $\alpha$ and $\beta$ be in $R$ such that
$\alpha \overline{\alpha} + \beta \overline{\beta} = -1$ and
$\alpha \overline{\beta} = \overline{\alpha}\beta $.
Let $\x_1$ and $\x_2$ be vectors in $R^{2\ell}$
such that
$\langle \x_1, \x_2 \rangle = 0$ in $R$ and
$\langle \x_i, \x_i \rangle = -1$ in $R$
for each $i=1, 2$.
For each $i, 1\leq i \leq k$, let
$s_i = -{\langle \r_i, \x_1 \rangle}$,
$t_i = -{\langle \r_i, \x_2 \rangle}$, and
$\y_i = (s_i, t_i, \alpha s_i+ \beta t_i, \beta s_i -\alpha t_i)$
be a vector of length $4$. Then the following matrix
\begin{equation*}
G=
\left[
  \begin{array}{cccc|c}
    1 & 0 & 0 & 0 &~~~~~~~~ \x_1 ~~~~~~~~ \\
    0 & 1 & 0 & 0 &~~~~~~~~ \x_2 ~~~~~~~~ \\
    \hline
      &\y_1&   &   &         \r_1          \\
      &\vdots &&  &        \vdots         \\
      &\y_{k}&& &         \r_{k}          \\
  \end{array}
\right]
\end{equation*}
generates a self-dual code $C$ over $R$ of length $2\ell+4$.
\end{enumerate}
\end{thm}

The following theorem shows that the converses of
Theorem~\ref{thm:construction1} hold for self-dual codes over $R$
with some restrictions. It can be proved in a similar way as in
~\cite{KL}, thus we omit the proof. The {\it rank} of a code $C$ means the minimum number of generators of $C$. The {\it free rank} of $C$ is defined to be the maximum of the ranks of free $R$-submodules of $C$.

\begin{thm}\label{thm:converse1}
\begin{enumerate}
\item Assume that char $(\F_q)=2$ or $q \equiv 1 \pmod 4$.\\
Any self-dual code $C$ over $R$ of length $2\ell+2$ with free rank at least two
is obtained from some self-dual code over $R$ of length $2\ell$
by the construction method in Theorem~\ref{thm:construction1} (i).

\item  Assume that $q \equiv 3 \pmod 4$ and $\ell$ is even.\\
Any self-dual code $C$ over $R$ of length $2\ell+4$ with free rank at least four
is obtained from some self-dual code over $R$ of length $2\ell$
by the construction method in Theorem~\ref{thm:construction1} (ii).
\end{enumerate}
\end{thm}

\

As seen in Theorem~\ref{thm:converse1}, there is some restriction
(i.e. minimum free rank) for the converses. In order to release this
restriction, in Theorem~\ref{thm:main} we find certain conditions of
$m$ and $q$ under which the converse is true without the
restriction. The following lemma is needed for the proof of Lemma~\ref{lemma:type} and
Theorem~\ref{thm:main}, and it finds the explicit criterion for $Y^m
-1$ to have exactly two irreducible factors over $\Fq[Y]$, and it
also characterizes the unit group of $R$.

\begin{lem}\label{lemma:irreducible-units}
\begin{enumerate}
\item \label{item:condition}
$Y^m -1$ has exactly two irreducible factors over $\Fq[Y]$ if and only if
$m$ is a prime $p$ and $q$ is a primitive element of $\F_p$.
\item Assume that the condition in~(\ref{item:condition}) holds.
Then the unit group $R^\ast$ of $R$ consists of $f(Y)$ in $\Fq[Y]$ of degree $\leq p-1$ such that $f(1) \in \Fq^\ast$ and $\Phi_p(Y) \nmid f(Y),$
where $\Phi_p(Y) = Y^{p-1}+ Y^{p-2} + \cdots + Y +1.$
 Equivalently, $f(Y)$ in $\Fq[Y]$ of degree $\leq p-1$  is not a unit in $R$ if and only if $Y-1 \mid f(Y)$ or
$\Phi_p(Y) \mid f(Y)$ in $\Fq[Y]$.
Hence we have $| R^\ast | = (q-1)(q^{p-1}-1).$

\item Assume that the condition in~(\ref{item:condition}) holds. Then the
ideal $\langle Y-1 \rangle$ of $R$ has cardinality $q^{p-1}$ and the ideal
$\langle \Phi_p(Y) \rangle$ of $R$ has cardinality $q$. That is,
$\dim_{\Fq}\langle \phi^{-1}(Y-1) \rangle = p-1$ and
$\dim_{\Fq}\langle \phi^{-1}(\Phi_p(Y) \rangle = 1$.
\end{enumerate}
\end{lem}
\begin{proof}
For (i), we note that a primitive $m$th root of unity $\z$ belongs
to some extension field of $\Fq$ as $(m, q)=1$. There exists a prime
divisor $p$ of $m$. If $p \neq m$ then
$Y^m-1=(Y-1)\Phi_p(Y)(\frac{Y^m-1}{Y^p-1})$ has at least three
irreducible factors over $\Fq$. Thus, if $Y^m -1$ has exactly two
irreducible factors over $\Fq[Y]$, then we should have $m=p$. If
$m=p$, then $\Phi_p(Y)$ is irreducible if and only if all the roots of
$\Phi_p(Y)$ are Galois conjugates over $\Fq$, or equivalently, $q$ is a
primitive element of $\F_p$. The other direction is obvious.


To show (ii), by the Chinese Remainder Theorem we have the following
canonical isomorphism
$$\psi : R \longrightarrow  \Fq[Y]/(Y-1) \oplus  \Fq[Y]/(\Phi_p(Y)).$$
Then $f(Y)$ is a unit of $R$ if and only if $\psi(f(Y))$ is a unit, equivalently,
$f(1) \in \Fq^\ast$ and $\Phi_p(Y) \nmid f(Y)$, so the result follows.

(iii) is clear.
\end{proof}


\

\begin{lem}\label{lemma:GeneralType}
Let $F_1$ and $F_2$ be finite fields, and consider a ring $\R = F_1 \times F_2$.
Let $e_i \in F_i^\times$ for $i = 1, 2$ and $f_1 = (e_1, 0), f_2 = (0, e_2) \in R$. Then
every linear code over $\R$ has a generator matrix (up to permutation equivalence) as follows:
\begin{equation} \label{eqn:16}
G=
\left[
  \begin{array}{ccccc}
    I_{k_1} & A_{12} & A_{13} & A_{14}  & A_{15} \\
    O       & f_1 I_{k_2} & f_2 M_{k_2}&  B_{24}  & B_{25} \\
    O       & O & O & \alpha I_{k_3}  &  \alpha D_{35}
  \end{array}
\right],
\end{equation}
where $\alpha \in \{f_1, f_2 \}$, $I_{k_i}$ is the $k_i \times k_i$ identity matrix $i=1, 2, 3$, \
$M_{k_2}$ is a $k_2 \times k_2$ diagonal matrix with elements in
the main diagonal not contained in $\R f_1$, and all the elements of $B_{24}$ and $B_{25}$
are $0$ or nonunits in $\R$.
\end{lem}
\begin{proof}
We note that $\R = F_1 \times F_2 = Rf_1 \oplus Rf_2$ is a commutative ring with unity $1_\R = (1,1)$, zero
$0_\R = (0, 0)$ and $f_1f_2 =0_\R$. In fact, the group $\R^\ast$ of units of $\R$ is
$\R - (\R f_1 \cup \R f_2) = F_1^\times \times F_2^\times$, there exist $r_1, r_2 \in \R$ such that $1_\R = r_1 f_1 + r_2 f_2$, and
$Rf_i = \langle f_i \rangle $ is a maximal ideal of $\R$ for $i = 1, 2$.

Let $G_0$ be a generator matrix for $C$. We first note that there are four possible cases
for each row of $G_0$. The first case is that a row contains a unit of $\R$, and the second one is
that a row has no units but it contains both a nonzero element in $\langle f_1 \rangle$
and a nonzero element in $\langle f_2 \rangle$.
The third case is that a row consists of only the elements
in $\langle f_1 \rangle$, and the last case is that a row contains
only the elements in $\langle f_2 \rangle$.
Below we transform $G_0$ into $G$ by column permutation and
elementary row operations.

We notice that $G_0$ can be transformed into $G_1$ such that the first $k_1$ rows
(respectively the first $k_1$ columns) of $G_1$ are equal to
the first $k_1$ rows (respectively the first $k_1$ columns) of $G$ in Eq.~(\ref{eqn:16}).
Deleting the first $k_1$ rows and the first $k_1$ columns of $G_1$,
we make $G_2$. We may assume that there is no unit component in $G_2$ (up to row equivalence);
otherwise we can increase $k_1$.

Now assume that the first row of $G_2$ is $(g_1f_1, \ g_2f_2, \ \dots)$
with $g_1 =(a_1, b_1)  \not \in \langle f_2 \rangle$ and $g_2= (a_2, b_2) \not \in \langle f_1 \rangle$.
Since $g_1 =(a_1, b_1) \not \in \langle f_2 \rangle$, we have $a_1 \neq 0$, that is, $a_1 \in F_1^\times$,
and similarly, $a_2 \in F_2^\times$.
Thus there exists ${\tilde{g_1}}=(a_1^{-1}, c_2)$ in $\R^\ast$ such that $g_1f_1 {\tilde{g_1}} = f_1$.
Multiplying the first row of $G_2$ by ${\tilde{g_1}}$, we may assume that
the first row of $G_2$ is $(f_1, \ {\tilde{g_2}} f_2, \ \dots)$ with
$\tilde{g_2} := {\tilde{g_1}}g_2 \not \in \langle f_1 \rangle$.

We claim that all the components of the first column of $G_2$ are in $\langle f_1 \rangle$.
Suppose $g=(a, b)$ is in the first column of $G_2$ with $g\not \in \langle f_1 \rangle$.
If $g\not \in \langle f_2 \rangle$, then $g$ is a unit, which is impossible. Thus,
$g \in \langle f_2 \rangle$. This leads to a unit component in $G_2$ (up to row equivalence).

We therefore may assume that all the components of the first column after $f_1$ are zero
by elementary row operations.
Likewise each component of the second column of $G_2$ is in $\langle f_2 \rangle$.
Suppose $G_2$ has the following form
\begin{equation*}
\left[
  \begin{array}{ccc}
    f_1     & \tilde{g_2} f_2 & \cdots\\
    0       & \tilde{g_2}' f_2 & \cdots\\
    \vdots  & \vdots  &
  \end{array}
\right]
\end{equation*}
for some $\tilde{g_2}=(a_2, b_2), \tilde{g_2}'= (a_2', b_2') \not \in \langle f_1 \rangle$, where
we have $b_2, b_2' \in F_2^\times$.
We add $(0, -b_2'/b_2)\times$(the first row of $G_2$) to the second row of $G_2$.
Then we have
\begin{equation*}
\left[
  \begin{array}{ccc}
    f_1     & \tilde{g_2} f_2 & \cdots\\
    0       & 0 & \cdots\\
    \vdots  & \vdots  &
  \end{array}
\right].
\end{equation*}
In this way, we may assume that the components of the second column after $f_2$
are all zero. Now assume that the second row of $G_2$ is $(0, 0, f_1, g_3 f_2, \dots)$
for some $g_3 \not \in \langle f_1 \rangle$.
In other words, $G_2$ has the following form
\begin{equation*}
\left[
  \begin{array}{ccccc}
    f_1     & \tilde{g_2} f_2 & \beta & \gamma &\cdots\\
    0       & 0 & f_1     &   g_3 f_2 & \cdots\\
    \vdots  & \vdots  & \vdots & \vdots&
  \end{array}
\right].
\end{equation*}
By the same reasoning as above, we may assume that $\beta = \gamma = 0$.
Repeating the above process, after some possible column changes, we may thus
assume that $G_2$ has the following form for some $k_2$.
\begin{equation*}
\left[
  \begin{array}{ccc}
    f_1 I_{k_2}     &  f_2 M_{k_2} & B\\
    O                &       O          & D\\
  \end{array}
\right].
\end{equation*}
The rest of the theorem follows in a similar way.
\end{proof}

\begin{lem}\label{lemma:type}
Let $m$ be a prime $p$ and $q$ be a primitive element of $\F_p$.
Then a linear code $C$ over the ring $R= \F_q[Y]/(Y^m -1)$ has a generator matrix
$G$ in the following form $($up to permutation equivalence$)$:
\begin{equation} \label{eqn:16}
G=
\left[
  \begin{array}{ccccc}
    I_{k_1} & A_{12} & A_{13} & A_{14}  & A_{15} \\
    O       & (Y-1)I_{k_2} & \Phi_p(Y)M_{k_2}&  B_{24}  & B_{25} \\
    O       & O & O & \alpha I_{k_3}  &  \alpha D_{35}
  \end{array}
\right],
\end{equation}
where $I_{k_i}$ is the $k_i \times k_i$ identity matrix $i=1, 2, 3$, \
$M_{k_2}$ is a $k_2 \times k_2$ diagonal matrix with nonzero elements in
the main diagonal over $\F_q$, \ all the elements of $B_{24}$ and $B_{25}$
are $0$ or nonunits, and $\alpha$ is $Y-1$ or $\Phi_p(Y)$.
\end{lem}
\begin{proof}
As $m$ is a prime $p$ and $q$ is a primitive element of $\F_p$, $Y^m-1$ has exactly two irreducible factors $Y-1$ and $\Phi_p(Y)$  by Lemma~\ref{lemma:irreducible-units} (i).
From Lemma~\ref{lemma:irreducible-units} and Lemma~\ref{lemma:GeneralType},
the result follows immediately.
\end{proof}

\

The following theorem shows
that the building-up construction is a complete method for
constructing all $\ell$-quasi-cyclic self-dual codes of length
$m\ell$ over $\mathbb F_q$ subject to certain conditions of $m$ and $q$.

\begin{thm} \label{thm:main}
Every self-dual code $C$ over $R = \F_q[Y]/(Y^m-1)$ of length
$2\ell+2$ can be obtained by the building-up construction given in
Theorem~\ref{thm:construction1} {\rm(}up to permutation
equivalence{\rm)}, provided that char $(\F_q)=2$ or $q \equiv 1
\pmod 4$, $m$ is a prime $p$, and $q$ is a primitive element of
$\F_p$.

Equivalently, every $\ell$-quasi-cyclic self-dual code of length
$m\ell$ over $\mathbb F_q$ can be obtained as the image under
$\phi^{-1}$ of a code over $R$ which is obtained by the building-up
construction subject to the same conditions of $m$ and $q$ as above.
\end{thm}

\begin{proof}
Let $C$ be a self-dual code of length $2\ell$ over $R$ with a
generator matrix of the form in~(\ref{eqn:16}).
Then we first show the following properties: \begin{enumerate}
\item $k_3 = 0$ and $k_1 + k_2 = \ell$,
\item $k_1 \geq 1$,
\item $k_1 \geq 2$ \ \ if \ $2\ell \geq 4$.
\end{enumerate}

\begin{enumerate}
\item
By the Chinese Remainder Theorem, we have
\begin{equation*}
R= \frac{\F_q[Y]}{(Y^p-1)} \cong  \frac{\F_q[Y]}{(Y-1)} \oplus \frac{\F_q[Y]}{\Phi_p(Y)} \cong \F_q \oplus \F_q^{p-1}.
\end{equation*}
Define $\Psi_1 : R \rightarrow \F_q$ and  $\Psi_2 : R \rightarrow \F_q^{p-1}$ as natural projections.
We extend $\Psi_1$ componentwise:
\begin{equation*}
\Psi_1 : M(R, m, n) \rightarrow M(\F_q, m, n),
\end{equation*}
where $M(R, m, n)$ and $M(\F_q, m, n)$ are the $m \times n$ matrix
spaces over $R$ and $\F_q$, respectively. Similarly we extend
$\Psi_2$. By Theorem 4.2 in~\cite{QCI}, $C = C_1 \oplus C_2$, where
$C_1$ is a self-dual code over $\F_q$ and $C_2$ is a self-dual
code over $\F_q^{p-1}$. By the proof of Theorem 6.1 in~\cite{QCIII},
$\Psi_1(G)$ and $\Psi_2(G)$ are generator matrices for $C_1$ and
$C_2$, respectively, so we have $\rank(\Psi_1(G)) = \ell =
\rank(\Psi_2(G))$. If $\alpha = Y-1$, then $\rank(\Phi_1(G)) =k_1 +
k_2$ and $\rank(\Phi_2(G)) =k_1 + k_2 + k_3,$ which shows that $k_3=
0$ and $k_1 + k_2 = \ell$. It is also shown similarly for the other
case $\alpha = \Phi_p(Y)$.
\item
We claim that there is a unit in the first component of some
codeword in $C$. Suppose there is no unit in the first component of
all codewords in $C$. Then we assume that all the first components
are in $\langle Y-1 \rangle$ or $\langle \Phi_p(Y) \rangle$.  This is because
the first components cannot contain both a nonzero element in
$\langle Y-1 \rangle$ and a nonzero element $\langle \Phi_p(Y)
\rangle$, since some $R$-linear combination of those two elements is a unit in
$R$ by  Lemma~\ref{lemma:irreducible-units} (ii).
If all the first components are in $\langle Y-1 \rangle$, then
$(\Phi_p(Y), 0, 0, \dots, 0)$ is in $C^{\perp} = C$ which is a
contradiction. Similarly, if all the first components are in
$\langle \Phi_p(Y) \rangle$, then $(Y-1, 0, 0, \dots, 0)$ is in
$C^{\perp} = C$ which is a contradiction. Therefore there is a unit
in the first component of some codeword in $C$. Hence $k_1 \geq 1$.

\item
From (i) and (ii) we have $k_1 \geq 1$ and $k_3 = 0$. We then
first observe that the column size of the last block of $G$ in
Eq.~(\ref{eqn:16}) is exactly $k_1$ as $k_1+k_2 = \ell.$ To get a
contradiction, suppose $k_1 = 1$. Then by Lemma~\ref{lemma:type},
$G$ is of the following form with $\gamma, \beta_i \in R (1\leq i
\leq \ell-1)$;
\begin{equation*}
G=
\left[
  \begin{array}{ccccc}
    1 & A_{12} & A_{13} & \gamma \\
    0 &        &        &  \beta_1        \\
    \vdots  & (Y-1)I_{k_2} & \Phi_p(Y)M_{k_2}& \vdots \\
    0  &        &        &  \beta_{\ell -1}        \\
  \end{array}
\right].
\end{equation*}
Let $\r_2$ be the second row of $G$ with $\r_2 =(0, Y-1, 0, \ldots,
0, c_1\Phi_p(Y), 0, \dots, 0, \beta_1)$ for some $c_1$ in
$\F_q^{\ast}$. Then
\begin{equation*}
0 = \langle \r_2, \r_2 \rangle = (Y-1)(\overline{Y}-1) + c_1^2 \Phi_p(Y)\overline{\Phi_p(Y)}+ \beta_1 \overline{\beta_1}.
\end{equation*}
But, we can see that $h(Y):= (Y-1)(\overline{Y}-1) + c_1^2 \Phi_p(Y)\overline{\Phi_p(Y)}= (2-Y-\overline{Y}) + c_1^2 \Phi_p(Y)\overline{\Phi_p(Y)}$
is a unit in $R$ by Lemma~\ref{lemma:irreducible-units}; in fact,
$h(1)= p^2 c_1^2 \in \Fq^\ast$ and $2-Y-\overline{Y} =
-(Y^{p-1}+Y-2)$ is not divisible by $\Phi_p(Y)$, and so $\Phi_p(Y)
\nmid h(Y).$ Therefore, $\beta_1\overline{\beta_1}$ is a unit in
$R,$ and hence $\beta_1$ is a unit. This is a contradiction because
$\beta_1$ is $0$ or a nonunit by Eq.~(\ref{eqn:16}). Therefore $k_1
\geq 2.$
\end{enumerate}

Now suppose that $C$ is a self-dual code over $R$ of length $2
\ell+2$. Then $k_1 \ge 2$ by (iii) above. Hence Eq.~(\ref{eqn:16})
gives a generator matrix in (i) of Theorem~\ref{thm:converse1}. Thus
it follows from (i) of Theorem~\ref{thm:converse1} that $C$ is
obtained from some self-dual code over $R$ of length $2 \ell$ by the
construction in  (i) of Theorem~\ref{thm:construction1}.
\end{proof}

What follows shows that in the binary cubic self-dual codes we can eliminate
the restriction for the converse of the construction in Theorem~\ref{thm:construction1}.
in other words, it shows that any binary cubic self-dual codes can be found by the building-up
construction in Theorem~\ref{thm:construction1}.

\begin{cor} \label{cor:07}
Let $R = \F_2[Y]/(Y^3 -1)$.
Let $C$ be a self-dual code over $R$ of length $2\ell+2$.
Then $C$ is obtained from some self-dual code over $R$ of length $2\ell$
by the construction method in Theorem~\ref{thm:construction1} (up to equivalence).
\end{cor}

\section{Construction of binary quasi-cyclic self-dual codes}
\label{sec:construction}

In this section we construct binary cubic quasi-cyclic self-dual codes and binary
quintic quasi-cyclic self-dual codes by using Theorem~\ref{thm:construction1}.

\subsection{Binary cubic self-dual codes}
\label{subsec:BinaryCubic}

A. Bonnecaze, et. al.~\cite{cubic} have studied binary cubic self-dual codes, and
they have given a partial list of binary cubic self-dual codes of lengths $\le 72$
by combining binary self-dual codes and Hermitian self-dual codes.

Using Corollary~\ref{cor:07}, we find a complete classification of
binary cubic self-dual codes of lengths up to $24$ (up to
permutation equivalence). To save space, we post the classification
up to $n=30$ in~\cite{web_Lee}.

We note that the classification of binary self-dual codes of lengths up to $32$ was given by
Pless and Sloane~\cite{PleSlo_75} and Conway, Pless and Sloane~\cite{ConPleSlo}; hence
it is possible to classify all binary cubic self-dual codes of length $32$.

Below is the summary.

\begin{thm}
Up to permutation equivalence,
\begin{enumerate}
\item there is a unique binary cubic self-dual code of length $6$.

\item there are
exactly two binary cubic self-dual codes of length $12$, one of which is
extremal.

\item there are exactly three binary cubic self-dual codes of length
$18$, one of which is extremal.

\item there are exactly sixteen binary cubic self-dual codes of length
$24$, where the extended Golay code and the odd Golay code of length
$24$ are obtained.

\end{enumerate}
\end{thm}

For even $\ell \ge 10$ we have tried to construct as many codes as
possible due to computational complexity. Recall that we have
summarized the number of extremal cubic self-dual codes of lengths
$\le 66$ in Table~\ref{tab:bin_cubic_sd}.

Using the following lemma, we determine possible weight enumerators of a binary
$\ell$-quasi-cyclic self-dual code of length $p\ell$ with $p$ a prime.

\begin{lem}\rm{(\cite[Ch.~16, Sec.~6]{MacSlo})} Let $C$ be a binary code and
$H$ any subgroup of $\mbox{Aut}(C)$. If $A_i$ is the total number of codewords
in $C$ of weight $i$, and $A_i(H)$ is the number of codewords which are fixed
by some non-identity element of $H$, then
\[
A_i \equiv A_i(H) \pmod{|H|}.
\]
\end{lem}

We remark that in \cite[Ch.~16, Sec.~6]{MacSlo} $A_i(H)$ is defined as the number of codewords which are fixed by some element of $H$.
Since the identity of $H$ always fixes any codeword, we need to
consider some non-identity element of $H$. Thus the codewords of weight $i$ can be divided into
two classes, those fixed by some {\em non-identity} element of $H$,
and the rest. Then just follow the proof of \cite[Ch.~16, Sec.~6]{MacSlo}.

\begin{cor} \label{cor:3_divide_Ai}
Let $C$ be a binary $\ell$-quasi-cyclic self-dual code of length
$p\ell$ with $p$ a prime. If the weight $i$ is not divisible by $p$,
then $A_i$ is divisible by $p$. In particular, $A_d$ is a multiple
of $p$ if $d$ is not divisible by $p$.
\end{cor}
\begin{proof}
We know from~\cite[Propositon A.1]{QCI} that if $p$ denotes a prime,
a binary code $C$ of length $\ell p$ is $\ell$-quasi-cyclic if and
only if $\mbox{Aut}(C)$ contains a fixed-point free (fpf)
permutation of order $p$. Hence $C$ contains an fpf permutation
$\sigma$ of order $p$. Let $H=\left <\sigma \right>$ whose order is
$p$. Since $\sigma$ is an fpf of order $p$ and any codeword of
weight $i$ with $p \nmid i$ cannot be fixed by any non-identity
element of $H$, we have $A_i(H) = 0$. Therefore by the above lemma,
$A_i \equiv 0 \pmod{p}$.
\end{proof}

\begin{enumerate}

\item $\ell = 10$, $[30, 15, 6]$ codes \\

There are three weight enumerators for self-dual $[30,15,6]$
codes~\cite{ConSlo}:

$W_1=1 + 19 y^6 + 393y^8 + 1848y^{10}+ 5192y^{12} + \cdots$.

$W_2=1+ 27y^6 + 369y^8 + 1848y^{10} + 5256y^{12}+ \cdots$.

$W_3=1+ 35y^6+ 345y^8 + 1848y^{10} + 5320y^{12} + \cdots$.

It is known~\cite{ConPleSlo},~\cite{ConSlo} that there are precisely
three codes with $W_1$, a unique code with $W_2$, and precisely nine
codes with $W_3$. Only two cubic self-dual $[30,15,6]$ codes are
given in~\cite{cubic}. We have constructed three codes with $W_1$
whose group orders are $576, 1152, 18432$ respectively. We have also
constructed five codes with $W_3$ whose group orders are $30, 192,
1440, 40320, 645120$. To save space, we post these codes
in~\cite{web_Lee}. In fact, these are all the cubic self-dual $[30,15,6]$ codes by the following calculation.

On the other hand, we have noticed that Munemasa has posted all
binary self-dual $[30,15]$ codes in~\cite{Mun2010Web}. Let $C_i$ be
the $i$th code in his list. By Magma, $C_i$ has $d=6$ if and only if
$i \in \{11, 61, 98, 119, 174, 184, 217, 350, 379, 397, 419, 487,
697\}$. We have further checked that the three codes with $W_1$
denoted by $C_{397}, C_{419}, C_{697}$ are all cubic and only five
out of the nine codes with $W_3$, denoted by $C_{119}, C_{174},
C_{184}, C_{350}, C_{487}$, are cubic. We have also checked that
there is no cubic code with $W_2$.

\begin{thm}
Up to permutation equivalence,
 there are exactly $8$ binary cubic self-dual $[30,15,6]$ codes.
\end{thm}

\item $\ell = 12$, \ $[36, 18, 8]$ codes \\

There are two weight enumerators for self-dual $[36,18,8]$ codes (refer to~\cite{ConSlo},~\cite{MelGal}):

$W_1=1+ 225y^8+ 2016y^{10} + \cdots$.

$W_2=1+ 289y^8 + 1632y^{10} + \cdots$.

For cubic codes, $p=3$ should divide $A_8$ by
Corollary~\ref{cor:3_divide_Ai}. Therefore any binary cubic
self-dual $[36,18,8]$ code has weight enumerator $W_1$. Bonnecaze, et.
al.~\cite{cubic} gave one code $CSD_{36}$ with $W_1$ and group order
$288$. We have found $9$ inequivalent cubic self-dual $[36,18,8]$
codes with $W_1$ and groups orders $18, 24, 36, 48, 96, 240, 288,
384,$ and $12960$. We have checked by Magma that our code with group
order $288$ is equivalent to $CSD_{36}$. Hence there are at least
$9$ extremal cubic self-dual codes of length $36$. These codes are posted in ~\cite{web_Lee}.

It is shown~\cite{MelGal} that there are exactly $41$ binary
self-dual $[36,18,8]$ codes and exactly $25$ codes among them have
$A_8=225$. However we have noticed that many generator matrices
in~\cite{MelGal} do not produce self-dual codes. This was confirmed
by Gaborit~\cite{Gab2011} and was corrected in his
website~\cite{Gab2011_Web}. From the corrected list of the binary
self-dual $[36,18,8]$ codes~\cite{Gab2011_Web}, we have checked that
only $13$ of the $25$ self-dual $[36,18,8]$ codes with $A_8=22$ are
cubic by further investigating the existence of a fixed point free
automorphism of order $3$ in each code. Let $C_i$ be the $i$th code
from the list of~\cite{Gab2011_Web}. Then $C_i$ is cubic if and only
if $i \in \{1, 3, 6, 7, 8, 9, 11, 12, 14, 16, 21, 22, 25 \}$.

Independently, Harada and Munemas~\cite{HarMun2010} have recently
classified all binary self-dual $[36,16]$ codes including the
extremal self-dual $[36,16,8]$ codes. They confirmed that there are
exactly $41$ extremal self-dual $[36,16,8]$ codes and exactly $25$
codes among them have $A_8=225$. Let $C_i$ be the $i$th code from
the list of~\cite{HarMun2010}. Then $C_i$ is cubic if and only if $i
\in \{1, 4, 12, 13, 15, 16, 19, 21, 24, 26, 27, 31, 33\}$.

\begin{thm}
Up to permutation equivalence,
 there are exactly $13$ binary cubic self-dual $[36,18,8]$ codes.
\end{thm}

\item $\ell = 14$, \ $[42,21,8]$ codes \\

There are two weight enumerators for self-dual $[42,21,8]$ codes~\cite{BHM, Huf}:

$W_1=1+164y^8 + 679y^{10} + \cdots$.

$W_2=1+(84 + 8 \beta)y^8+ (1449-24 \beta)y^{10} + \cdots$~ $(\beta \in \{0, 1, . . . , 22, 24, 26, 28, 32, 42\})$.

By Corollary~\ref{cor:3_divide_Ai}, $3$ should divide $A_8$.
Therefore any binary cubic self-dual $[42,21,8]$ code has weight
enumerator $W_2$, where $3$ divides $84+ 8 \beta$, that is, $\beta$
is a multiple of $3$. Bonnecaze, et. al.~\cite{cubic} gave one code
with $W_2$ and $\beta=0$. We have found $14$ inequivalent cubic
self-dual $[42,21,8]$ codes with $\beta=0,3,6, 9, 12$ with group
orders $3,6, 12,$ and $36$. It is shown that if a self-dual code satisfies $W_2$ with
$\beta \in \{24, 26, 28, 32, 42\}$, it is equivalent to one of the
eight codes in~\cite[Table 1]{BHM}. If it is cubic, then $\beta$
should be $\beta=24$ or $42$ by the divisibility condition on
$\beta$. For $\beta=24$, there are three codes denoted by $C_{24,1},
C_{24,2}, C_{24,3}$~\cite{BHM}. We have checked that only $C_{24,2}$
has a fixed point free automorphism of order $3$; hence it is cubic.
For $\beta=42$, there is only one code denoted by
$C_{42}$~\cite{BHM}. We have checked that it has a fixed point free
automorphism of order $3$; hence it is cubic.

We have found~\cite[Table 5]{BouYanRus} where it is shown that there are exactly
$1569$ binary self-dual $[42,21,8]$ codes with a fixed point free automorphism of order $3$ and weight enumerator $W_2$. This table confirms the above calculations.

\begin{thm}
Up to permutation equivalence,
 there are exactly $1569$ binary cubic self-dual $[42,21,8]$ codes.
\end{thm}

\item $\ell = 16$, \ $[48,24,10]$ codes \\

There are two weight enumerators for self-dual $[48,24,10]$ codes~\cite{Huf}:

$W_1=1+704y^{10} + 8976y^{12} + \cdots$.

$W_2=1+ 768y^{10} + 8592y^{12} + \cdots$.

By Corollary~\ref{cor:3_divide_Ai}, any binary cubic self-dual $[48,24,10]$ code has weight enumerator $W_2$. Bonnecaze, et. al.~\cite{cubic} gave one code with $W_2$ with no group order given. We have found four inequivalent codes with $W_2$ and group orders $3,6, 12,$ and $24$.
See Table~\ref{tab:cubic_48_54_66} for details,
where the first column gives the code name, the second and third columns the $X$ vector and the base matrix in Theorem~\ref{thm:construction1}, the fourth column the corresponding weight enumerator of the binary code, and the last column the order of the automorphism group of the binary code.

\item $\ell = 18$, \ $[54,27,10]$ codes \\

There are two weight enumerators for self-dual $[48,24,10]$ codes~\cite{Huf}:

$W_1=1 +(351-8\beta)y^{10} + (5031+ 24\beta)y^{12}+ \cdots$~ $(0 \le \beta \le 43)$.

$W_2=1+(351-8 \beta)y^{10} + (5543+24\beta)y^{12} + (43884+ 32\beta)y^{14} + \cdots$~ $(12 \le \beta \le 43)$.

Any binary cubic self-dual $[54,27,10]$ code has $W_1$ or $W_2$ as its weight enumerator; in both cases, $3$ divides $\beta$ with the same reasoning as above.
Bonnecaze, et. al.~\cite{cubic} gave two codes, one with $W_1$ and $\beta=0$ and the other with $W_2$ and $\beta=12$ (and group order $3$). We have found four inequivalent codes with $W_1$ and $\beta=0,3,6,9$ (all group orders $3$) and three inequivalent codes with $W_2$ and $\beta=12, 15, 18$ (all group orders $3$).
See Table~\ref{tab:cubic_48_54_66} for more details.

\item $\ell = 20$ \\
We have not found any self-dual $[60,30, 12]$ codes even though there are at least three cubic self-dual $[60,30, 12]$ codes~\cite{cubic} with $W_2$ and $\beta=10$ in the notation of~\cite{Huf}.

\item $\ell = 22$, \ $[66, 33, 12]$ codes \\
There are three possible weight enumerators for self-dual $[66, 33, 12]$ codes~\cite{Huf}:

$W_1 = 1 + 1690y^{12} + 7990y^{14} + \cdots,$

$W_2 = 1 + (858 + 8\beta)y^{12} + (18678 - 24\beta)y^{14} + \cdots$
($0 \le \beta \le 778$), and

$W_3 = 1 + (858 + 8 \beta)y^{12} + (18166 - 24\beta)y^{14} + \cdots$
($14 \le \beta \le 756$).

By Corollary~\ref{cor:3_divide_Ai}, any binary cubic self-dual
$[66,33,12]$ code should have weight enumerator $W_2$ with $\beta$
in the given range as above since $A_{14}$ should be divisible by
$3$. Bonnecaze, et. al.~\cite{cubic} gave two codes with $W_2$ and
$\beta=21, 30$. Using $G_{20}$ with various values of $X$ in
Table~\ref{tab:cubic_48_54_66}, we have constructed five
inequivalent codes with $W_2$ and $\beta=17, 23, 26, 43, 46$. All
have automorphism group of order $3$.
\end{enumerate}

The following generator matrices $G_{14}$, $G_{16}$, and $G_{20}$ are used in Table~\ref{tab:cubic_48_54_66}
for constructing binary extremal cubic self-dual codes of $n=48, 54, 66$. \\

$G_{14}=$ {\tiny
 $\left(
 \begin{array}{l}
   1,   0,   Y^2 + Y ,  Y + 1,   1,   Y,   Y^2 + Y + 1,   Y,   Y,   0,   Y^2 + Y,   Y^2,   1,  0\\
   Y,   Y,   1,   0,   Y^2,   Y^2 + Y + 1,   Y^2 + Y,   Y^2 + Y + 1,   Y^2 + Y + 1,   Y^2 +1,   Y^2,   Y + 1,   Y,   Y\\
   Y^2 + 1,   Y^2 + 1,   0,   0,   1,   0,   0,   Y^2 + Y,   Y,   Y^2,   Y^2 + 1,   Y,   Y^2 +Y + 1,   Y^2 + Y + 1\\
   1,   1,   Y^2 + 1,   Y^2 + 1,   Y^2 + 1,   Y^2 + 1,   1,   0,   Y^2,   Y^2 + Y,   Y^2,   Y^2 + Y + 1,   Y^2 + 1,   Y^2 + Y\\
   1,   1,   1,   1,   0,   0,   Y^2 + Y + 1,   Y^2 + Y + 1,   1,   0,   Y + 1,   Y^2 + 1,   Y^2 + Y,   Y^2 + Y + 1\\
   Y^2 + Y + 1,   Y^2 + Y + 1,   Y,   Y,   1,   1,   Y + 1,   Y + 1,   1,   1,   1,   0,   Y +1,   Y^2 + Y + 1\\
   Y^2,   Y^2,   1,   1,   Y^2 + Y + 1,   Y^2 + Y + 1,   Y^2,   Y^2,   Y^2,   Y^2,   Y,   Y,  1,   1\\
\end{array}
\right)$ }
\\

$G_{16}=$ {\tiny
 $\left(
 \begin{array}{l}
   1,   0,   Y^2 + Y,   0,   Y^2,   Y^2,   Y^2 + Y + 1,   Y + 1,  1,   Y,   Y^2,   Y^2,   1,  Y^2,   Y + 1,  0\\
   Y + 1,   Y + 1, 1,   0,   Y^2,   Y,   1,   Y^2,   Y^2 + 1,   1,   Y,   Y^2,   Y^2 + Y + 1,   Y,   Y + 1,   Y +1\\
   Y,   Y,   Y + 1,   Y + 1, 1,   0,   Y^2,   Y^2 + Y + 1,   Y^2 + Y,   Y^2 + Y + 1,   Y^2 + Y + 1,   Y^2 + 1,Y^2,  Y + 1,   Y,   Y\\
   Y^2 + Y,  Y^2 + Y,  Y^2 + Y,   Y^2 + Y,  0,   0,     1,   0,   0,   Y^2 + Y,   Y,   Y^2,  Y^2 + 1,   Y,   Y^2 + Y+ 1,   Y^2 + Y + 1\\
   0,   0, 1,   1,  Y^2 + 1,  Y^2 + 1,  Y^2 + 1,   Y^2 + 1,   1,   0,   Y^2,   Y^2 + Y,   Y^2,   Y^2 + Y + 1,   Y^2+ 1,   Y^2 + Y\\
   1,   1, Y^2 + Y,   Y^2 + Y, 1,   1,  0,   0,   Y^2 + Y + 1,   Y^2 + Y + 1,   1,   0,   Y + 1,   Y^2 + 1,   Y^2 +Y,   Y^2 + Y + 1\\
   Y + 1,   Y + 1, Y^2 + Y + 1,   Y^2 + Y + 1,  Y,   Y,  1,   1,   Y + 1,   Y + 1,   1,   1,   1,   0,   Y + 1,Y^2 + Y + 1\\
   Y^2 + Y + 1,   Y^2 + Y + 1,  Y,   Y,  1,   1,  Y^2 + Y + 1,   Y^2 + Y + 1,   Y^2,   Y^2,   Y^2,   Y^2,   Y,   Y, 1,   1\\
\end{array}
\right)$ }

$G_{20}=$ {\tiny
 $\left(
 \begin{array}{l}
    1, 0, 0 , Y^2 + 1 , Y^2 + Y + 1 , Y , Y^2 , Y^2 , 1 , Y^2 + Y,0 , 1 , Y + 1 , 1 , Y^2 , Y , Y + 1 , Y^2 + 1 , 1 , 1\\
    Y, Y, 1 , 0 , Y + 1 , Y + 1 , Y + 1 , 1 , Y + 1 , 1,Y^2 + Y + 1 , Y , Y^2 , Y^2 + Y , Y^2 , 1 , Y + 1 , Y^2 , 1 , Y + 1\\
    Y^2+Y+1, Y^2+Y+1, Y^2 + 1 , Y^2 + 1 , 1 , 0 , Y^2 + Y , 0 , Y^2 , Y^2, Y^2 + Y + 1 , Y + 1 , 1 , Y , Y^2 , Y^2 , 1 , Y^2 , Y + 1 , 0\\
    0, 0, 0 , 0 , Y + 1 , Y + 1 , 1 , 0 , Y^2 , Y ,1 , Y^2 , Y^2 + 1 , 1 , Y , Y^2 , Y^2 + Y + 1 , Y , Y + 1 , Y + 1\\
    Y+1, Y+1, Y^2 + Y , Y^2 + Y , Y , Y , Y + 1 , Y + 1 , 1 , 0,Y^2 , Y^2 +Y + 1 , Y^2 + Y , Y^2 + Y + 1 , Y^2 + Y + 1 , Y^2 + 1 , Y^2 , Y + 1 , Y , Y\\
    0, 0, Y , Y , Y^2 + Y , Y^2 + Y , Y^2 + Y , Y^2 + Y , 0 , 0 ,1 , 0 , 0 , Y^2+ Y , Y , Y^2 , Y^2 + 1 , Y , Y^2 + Y + 1 , Y^2 + Y + 1\\
    Y^2+Y+1, Y^2+Y+1, 0 , 0 , 0 , 0 , 1 , 1 , Y^2 + 1 , Y^2 + 1,Y^2 + 1 , Y^2 + 1 , 1 , 0 , Y^2 , Y^2 + Y , Y^2 , Y^2 + Y + 1 , Y^2 + 1 , Y^2 + Y\\
    Y+1, Y+1, 1 , 1 , 1 , 1 , Y^2 + Y , Y^2 + Y , 1 , 1,0 , 0 , Y^2 + Y + 1 , Y^2 + Y + 1 , 1 , 0 , Y + 1 , Y^2 + 1 , Y^2 + Y , Y^2 + Y + 1\\
    Y^2 + Y + 1,   Y^2 + Y + 1, Y^2 + 1 , Y^2 + 1 , Y + 1 , Y + 1 , Y^2 + Y + 1 , Y^2 + Y + 1 , Y , Y,1 , 1 , Y + 1 , Y + 1 , 1 , 1 , 1 , 0 , Y + 1 , Y^2 + Y + 1\\
    Y^2 + Y + 1,   Y^2 + Y + 1, 1 , 1 , Y^2 + Y + 1 , Y^2 + Y + 1 , Y , Y , 1 , 1,Y^2 + Y + 1 , Y^2 + Y + 1 , Y^2 , Y^2 , Y^2 , Y^2 , Y , Y , 1 , 1\\
\end{array}
\right)$
}

\begin{table}
\centering \caption{Binary extremal Type I cubic self-dual codes of length $n=48,54,66$}
\label{tab:cubic_48_54_66}
{\tiny{
\[
\begin{tabular}{|c|c|c|c|c|}
\hline

\hline Codes $C_{n,i}$ & $X$ vector & Using Gen & Weight & $|\mbox{Aut}|$ \\
& & Matrix & Enumerator & \\
\hline
$C_{48,1}$ & $(Y, Y + 1, Y^2 + 1, Y^2, 0, 1, 0,$ & $G_{14}$ & $W_2$& 3 \\
& $Y^2, 0, Y^2 + Y, 0,  Y^2 + Y,  Y^2,  0)$ & &&\\
\hline
$C_{48,2}$ & $(Y^2 + Y, 0, Y^2 + Y + 1, 1, Y, 1, Y + 1,$ & $G_{14}$ & $W_2$& 24 \\
& $Y^2 + 1, Y^2 + 1, 1, Y^2 + Y + 1, 1, Y^2, Y^2)$ && & \\
\hline
$C_{48,3}$ & $(Y^2, Y^2 + Y + 1, Y^2, 0, Y, 0, Y^2 + Y + 1,$ & $G_{14}$ & $W_2$ & 12 \\
& $Y^2 + Y, 0, Y + 1, Y, Y^2 + 1, Y, Y^2 + 1)$ & && \\
\hline
$C_{48,4}$& $(0, 0, Y^2 + Y, Y^2 + Y + 1, Y + 1, Y,$ & $G_{14}$ & $W_2$ & 6 \\
& $1, Y^2 + Y, Y^2 + 1, Y^2, Y + 1, Y^2, Y, 1)$ & & &\\
\hline
$C_{54,1}$ & $(Y^2 + Y + 1, Y^2 + 1, Y^2+1, Y^2+1,
Y^2 + 1, Y^2 + Y, Y^2 + Y + 1,$ & $G_{16}$ &  $W_2$, $\beta=18$ & 3\\
&$Y^2 + Y + 1, Y^2 + Y, Y^2, 0, Y + 1, 1, 0, Y^2 + Y + 1, Y^2 + Y + 1)$  &&&\\
\hline
$C_{54,2}$& $(Y + 1, Y + 1, Y + 1, 1, Y + 1, 1, Y^2 + Y + 1, Y, Y^2,$ & $G_{16}$ & $W_1$, $\beta=9$ & 3 \\
& $Y^2 + Y, Y^2, 1, Y + 1, Y^2, 1, Y + 1)$ & & &\\
\hline
$C_{54,3}$ & $(Y, Y^2, Y + 1, 0, 1, Y^2, Y, Y^2 + 1, 1,  Y^2 + Y, 1, Y,$ & $G_{16}$ & $W_2$, $\beta=15$ & 3 \\
& $Y^2 + Y + 1, 1, Y^2 + Y + 1, Y^2 + 1)$ & & &\\
\hline
$C_{54,4}$& $(Y^2 + Y, Y^2 + Y + 1, Y^2 + Y, 1, Y^2 + 1, Y + 1, 0, Y^2 + Y,$ & $G_{16}$  &$W_1$, $\beta=3$ & 3\\
& $Y^2, 1, 1, 0, Y^2 + 1, Y, 1, Y^2 + 1)$ & & &\\
\hline
$C_{54,5}$& $(1, Y, Y, Y, Y + 1, Y^2, Y, 0, Y + 1, Y^2 + Y, Y^2,$ & $G_{16}$ & $W_1$, $\beta=0$ & 3 \\
& $Y^2 + Y + 1, Y^2, Y^2 + Y, 0, Y + 1)$ & & &\\
\hline
$C_{54,6}$& $(Y^2, 0, Y^2, Y^2 + Y + 1, Y^2 + Y, 0, 0,$ & $G_{16}$ & $W_2$, $\beta=12$& 3 \\
& $Y^2 + 1, 0, Y^2 + Y, Y, 0, Y^2, Y^2 + Y, Y + 1, 0)$ & & &\\
\hline
$C_{54,7}$ & $(Y^2 + Y + 1, Y^2 + Y, Y^2 + Y, Y + 1, Y, Y^2, Y^2 + Y, Y^2 + Y + 1,$ & $G_{16}$ & $W_1$, $\beta=6$ &3 \\
& $Y^2 + Y + 1, Y^2 + 1, Y^2 + Y, Y^2, Y^2 + 1, Y^2 + Y + 1, Y + 1, 0)$ && &\\
\hline
$C_{66,1}$ & $(Y^2+1, 1, Y+1, 1, 0, 0, Y^2+Y+1, 0, 1, Y^2$,  & $G_{20}$ & $W_2$, $\beta=46$ & 3 \\
& $1, Y, Y+1, 1, 1, Y^2+Y, 0, Y+1, 0, 0)$ && &\\
\hline
$C_{66,2}$ & $(Y^2+Y+1, Y^2+Y+1, 0, 1, Y, Y^2, Y^2, 1, Y^2+Y+1, Y^2+Y+1,$ &$G_{20}$ & $W_2$, $\beta=17$ & 3 \\
& $Y^2+1, Y^2, Y^2+1, 0, Y^2+Y+1, Y^2+Y+1, Y^2+1, 0, Y, Y+1)$ &&&\\
\hline
$C_{66,3}$& $(0,0,Y^2+Y,1, Y^2+Y, Y^2+Y+1, Y+1, 1, Y+1, Y, Y^2+Y+1,$ &$G_{20}$ & $W_2$, $\beta=23$ & 3 \\
& $Y, Y^2+1, Y+1, Y^2, Y+1, Y+1, Y^2+Y+1, Y, Y+1)$ &&&\\
\hline
$C_{66,4}$ & $(Y^2, Y^2+1, Y^2, Y^2, Y+1, 0, 1, 0, 1, Y^2+1, Y^2+1,$ & $G_{20}$ &$W_2$, $\beta=26$ & 3 \\
& $1, Y^2+Y, Y+1, 1, Y, Y+1, Y^2+1, 0, Y^2)$ &&& \\
\hline
$C_{66,5}$&$(Y, Y, Y^2, Y^2+1, Y+1, Y, 0, Y+1, Y^2+Y+1, 0, Y^2+1,$&  $G_{20}$ & $W_2$, $\beta=43$ & 3 \\
& $Y^2+Y+1, 1, Y, Y^2+Y+1, Y^2+Y, 0, Y^2+1, Y^2+Y, 0)$ & &&\\
\hline

\hline
\end{tabular}
\]
}}
\end{table}

\comment{

\begin{rem}
It is reasonable to stop at length $24$. If we want to classify cubic elf-dual codes of length $30$,
then we should construct all self-dual codes over $R$ of length $10$ from he self-dual codes which are constructed
from $1920$ self-dual codes over $R$ of length $6$.
We think that it takes too long.
\end{rem}

We have made random construction from $G_2 = [1~ 1]$ and
obtained $G_4, G_6, \cdots, G_{16}$ successively using building-up construction with $c=1$.
The following matrices, $L$ and $R$, are the left half and the right half of $G_{16}$.
Namely, $G_{16} = [L ~|~ R]$.
{\scriptsize{
\begin{equation*}
L = \left[
  \begin{array}{cccccccc}
1&   0&   Y^2 + Y&   0&   Y^2&   Y^2&   Y^2 + Y + 1&   Y + 1 \\
Y + 1&   Y + 1& 1&   0&   Y^2&   Y&   1&   Y^2 \\
Y&   Y&   Y + 1&   Y + 1& 1&   0&   Y^2&   Y^2 + Y + 1 \\
Y^2 + Y&  Y^2 + Y&  Y^2 + Y&   Y^2 + Y&  0&   0&     1&   0 \\
0&   0& 1&   1&  Y^2 + 1&  Y^2 + 1&  Y^2 + 1&   Y^2 + 1 \\
1&   1& Y^2 + Y&   Y^2 + Y& 1&   1&  0&   0 \\
Y + 1&   Y + 1& Y^2 + Y + 1&   Y^2 + Y + 1&  Y&   Y&  1&   1 \\
Y^2 + Y + 1&   Y^2 + Y + 1&  Y&   Y&  1&   1&  Y^2 + Y + 1&   Y^2 + Y + 1
  \end{array}
\right],
\end{equation*}
}}
{\scriptsize{
\begin{equation*}
R = \left[
  \begin{array}{cccccccc}
1&   Y&   Y^2&   Y^2&   1&  Y^2&   Y + 1&  0 \\
  Y^2 + 1&   1&   Y&   Y^2&   Y^2 + Y + 1&   Y&   Y + 1&   Y + 1 \\
  Y^2 + Y&   Y^2 + Y + 1&   Y^2 + Y + 1&   Y^2 + 1&   Y^2&  Y + 1&   Y&   Y \\
 0&   Y^2 + Y&   Y&   Y^2&  Y^2 + 1&   Y&   Y^2 + Y + 1&   Y^2 + Y + 1 \\
  1&   0&   Y^2&   Y^2 + Y&   Y^2&   Y^2 + Y + 1&   Y^2 + 1&   Y^2 + Y \\
  Y^2 + Y + 1&   Y^2 + Y + 1&   1&   0&   Y + 1&   Y^2 + 1&   Y^2 + Y&   Y^2 + Y + 1 \\
    Y + 1&   Y + 1&   1&   1&   1&   0&   Y + 1&   Y^2 + Y + 1 \\
    Y^2&   Y^2&   Y^2&   Y^2&   Y&   Y&   1&   1
  \end{array}
\right].
\end{equation*}
}}
The corresponding binary cubic codes of $G_{10}, G_{12}, G_{14}$, and $G_{16}$
are $[30, 15, 6]$, $[36, 18, 8]$, $[42, 21, 8]$, and $[48, 24, 10]$ Type I codes
and all are optimal Type I codes. The weight enumerators are respectively
{\scriptsize{
\begin{eqnarray*}
&&1 + 19y^6 + 393y^8 + 1848y^{10} + 5192y^{12} + 8931y^{14} + \cdots, \\
&&1 + 225y^8 + 2016y^{10} + 9555y^{12} + 28800y^{14} + 55755y^{16} + 69440y^{18} + \cdots, \\
&&1 + 84y^8 + 1449y^{10} + 10640y^{12} + 50256y^{14} + 158718y^{16} + 337540y^{18} + 489888y^{20} +  \cdots, \\
&&1 + 768y^{10} + 8592y^{12} + 57600y^{14} + 267831y^{16} +
    871168y^{18} + 1997040y^{20} + 3264768y^{22} + 3841680y^{24}  + \cdots.
\end{eqnarray*}
}}
Note that
there are three different weight enumerators for $[30, 15, 6]$ Type I codes.
There are two different weight enumerators for $[36, 18, 8]$ Type I codes.
There are two possible forms for the weight enumerators of $[42, 21, 8]$ Type I codes.
Our $[42, 21, 8]$ Type I code corresponds to $W_2$ with $\beta =0$ in~\cite{Huf}.
There are two possible weight weight enumerators for $[48, 24, 10]$ Type I codes.
Our $[48, 24, 10]$ Type I code corresponds to $W_2$ in~\cite{Huf}.
} 

\comment{
\begin{rem}
It is easy to see that if two linear codes $C_1$ and $C_2$ over $R$ are equivalent,
then the corresponding quasi-cyclic codes $\phi^{-1}(C_1)$ and $\phi^{-1}(C_2)$ are equivalent.
However we have no answer for the converse. We guess that the converse may be false.
\end{rem}
}

\comment{
\begin{rem}
We wanted to find all inequivalent codes over $R$.
But we have no software tools for checking equivalence of codes over $R$.
Instead, we gave all inequivalent QC codes over $\F_2$ using Magma.
We know that there is one to one correspondence between a linear code over $R$
and a QC code over $\F_2$.
And it is clear that if two linear codes $C_1$ and $C_2$ over $R$ are equivalent,
then the corresponding QC codes $\phi^{-1}(C_1)$ and $\phi^{-1}(C_2)$ are equivalent.
But we don't know the answer for the following question. If two QC codes $C_1$ and $C_2$ over $\F_2$
are equivalent, then are $\phi(C_1)$ and $\phi(C_2)$ equivalent over $R$?
If the answer for the question is yes,
then the classification over $R$ and the classification over $\F_2$ are the same.
But we think that the answer is no.
\end{rem}

} 

\subsection{Binary quintic self-dual codes}
\label{subsec:class-quintic}

In this subsection, we give the classification of binary quintic
self-dual codes of even lengths up to $30$ (up to permutation
equivalence) by using Theorem~\ref{thm:main} since $2$ is a
primitive element of $\F_5$. Using the known classification of binary self-dual codes of lengths up to $30$, one can also classify binary quintic self-dual codes of these lengths.
To save space, we post the classification result in~\cite{web_Lee}.
We know from~\cite[Table F]{ConPleSlo} that there are exactly $13$
optimal binary self-dual $[30, 15, 6]$ codes with three distinct
weight enumerators $W_1, W_2, W_3$ from
Section~\ref{subsec:BinaryCubic}. Exactly nine of them have the
weight enumerator $W_3=1+ 35y^6+ 345y^8 + 1848y^{10} + 5320y^{12} +
\cdots$. By Corollary~\ref{cor:3_divide_Ai}, $W_3$ is the only
possible weight enumerator for a binary extremal quintic self-dual
code. We have checked that only four codes are binary quintic optimal self-dual codes of length $30$.

\begin{thm}
Up to permutation equivalence,
\begin{enumerate}
\item there is a unique quintic self-dual code of length $10$.
\item there are
exactly three quintic self-dual codes of length $20$, two of which are
extremal.
\item there are exactly eleven quintic self-dual codes of length
$30$, four of which are optimal.
\end{enumerate}
\end{thm}

Making successive random choices of ${\bf{x}}$ from $G_{6,2}$ by using the building-up construction in Theorem~\ref{thm:construction1} with $c=1$,
we obtain $G_{12} = [L ~|~ R]$, where $L$ and $R$ are given below.

{\tiny{
\begin{equation*}
L = \left[
  \begin{array}{cccccc}
1&   0&   Y^4 + Y^2 + Y&   Y^4 + Y^3 + Y^2 + 1&   Y^4 + Y^3 + Y^2&   Y^3 + Y \\

Y^4 + Y^2 + Y&   Y^4 + Y^2 + Y&  1&   0&   Y^4 + Y^2&   Y^3 + Y + 1 \\
Y^4 + Y^3 + Y^2 + Y + 1&   Y^4 + Y^3 + Y^2 + Y + 1&  Y^4 + Y^3 + Y + 1&
Y^4 + Y^3 + Y + 1&    1&   0 \\
Y^4 + Y^2&   Y^4 + Y^2&  1&   1&  Y^4&   Y^4 \\
Y^4 + Y^2 + 1&   Y^4 + Y^2 + 1&  Y^3 + 1&   Y^3 + 1&  Y^4 + Y^3 + Y^2 + Y&   Y^4 + Y^3 + Y^2 + Y \\
Y^4 + Y^2 + Y&   Y^4 + Y^2 + Y& Y^3 + Y^2 + 1&   Y^3 + Y^2 + 1&  Y^4 + Y^2 + 1&   Y^4 + Y^2 + 1
  \end{array}
\right].
\end{equation*}
}}

{\tiny{
\begin{equation*}
R = \left[
  \begin{array}{cccccc}
Y^4 + Y^3 + Y&   Y^4 + Y^2 + Y&   Y^4 + 1&   Y^3 + Y^2 + Y&   Y^4 + Y^2 + Y&   Y  \\
   Y^2 + Y& Y^4 + Y^3 + Y^2 + Y&   Y^4 + Y^3 + Y^2 + Y&   Y^2 + Y&   Y^4 + Y^3&   Y^4 + Y^2\\
   Y^4 + Y^3 + Y^2& Y^3 + Y&   Y&   Y^2&   Y^3 + Y&   Y^4 + Y \\
               1&   0&             0&       0&     Y + 1&   Y^3 + Y + 1 \\
   Y^4 + Y^2 + 1&     Y^4 + Y^2 + 1&   1&   0&   0&   1 \\
   Y^2&     Y^2&   1&   1&   1&   1
  \end{array}
\right].
\end{equation*}
}}

We verify that the corresponding binary quintic self-dual code of
$G_{12}$ has parameters $[60, 30, 12]$. The deletion of the first
two columns and the first row of $G_{12}$ is denoted by $G_{10}$,
and similarly we obtain $G_8$ from $G_{10}$. Their corresponding
binary quintic self-dual codes have parameters $[40, 20, 8]$ (Type
II) and $[50, 25, 10]$. We summarize their corresponding weight
enumerators of $G_8, G_{10}, G_{12}$ respectively as follows.

{\scriptsize{
\begin{eqnarray*}
&&1 + 285y^{8} + 21280y^{12} + 239970y^{16} + 525504y^{20} + \cdots, \\
&&1 + 516y^{10} + 7720y^{12} + 55880y^{14} + 291990y^{16} +
    1077265y^{18} + 2810424y^{20} + 5287640y^{22} + 7245780y^{24} + \cdots, \\
&&1 + 3195y^{12} + 29760y^{14} + 284625y^{16} + 1728000y^{18} +
    7769400y^{20} + 26392320y^{22} + 67226760y^{24} + 130060800y^{26} \\
&&   \hspace{0.2cm}  + 193151475y^{28} + 220449152y^{30} + \cdots.
\end{eqnarray*}
}}

The first one is the unique extremal weight enumerator,
the second weight enumerator corresponds to $W_2$ with $\beta =2$ in~\cite{Huf}, and
the third weight enumerator corresponds to $W_2$ with $\beta =10$ in~\cite{Huf}.
The orders of the automorphism groups are $10$, $5$, and $20$ respectively.

\section{Construction of quasi-cyclic self-dual codes over various finite fields}
\label{sec:construction-various}

In this section we find quasi-cyclic self-dual codes over $\mathbb F_2$, $\mathbb F_3$, $\mathbb F_4$ and $\mathbb F_5$ which are optimal or have best known self-dual codes by applying the building-up construction in Theorem~\ref{thm:construction1}.

 \subsection{Cubic self-dual codes over $\mathbb F_4$ and $\mathbb F_5$}
\label{susec:cubic-various}

In~\cite{HanKimLeeLee}, we have given cubic self-dual codes over $\mathbb F_4$ and $\mathbb F_5$ that are optimal or have best known parameters. In particular, we have the following.

\begin{thm}
There are at least two monomially inequivalent $[24, 12, 9]$
self-dual codes over $\F_5$, one of which is cubic and denoted by
$CSD_{24}^5$.
\end{thm}

Applying Construction $A$~\cite{ConSlo_sphere}, we can construct
the odd Leech lattice $O_{24}$ using the idea in~\cite{HarKha}.
In~\cite[Prop. 4]{HarKha} it is shown that
for a self-dual $[24,12, d\ge 8]$ code $C$ over $\F_5$, the
corresponding lattice $A_5(C)$ by Construction $A$ is the odd Leech
lattice $O_{24}$ if there is no codeword $\bf{x} \in C$ with
$n_0({\bf{x}})=14$, $n_1({\bf{x}})=10$, and $n_2({\bf{x}})=0$, where
$n_i({\bf{x}})$ denotes the number of coordinates of ${\bf{x}}$ with
$\pm i$ for $i=0,1,2$. We have calculated the complete weight
enumerator of $CSD_{24}^5$ by Magma and checked that there is no
such ${\bf{x}}$ in $CSD_{24}^5$. Thus $A_5(CSD_{24}^5)=O_{24}.$
Since it is known~\cite{ConSlo_sphere} that one of the two even
unimodular neighbors of $O_{24}$ is the Leech lattice
$\Lambda_{24}$, we have another way to construct $\Lambda_{24}$
using our new code $CSD_{24}^5$, rather than $\mathbb Q_{24}$ used
in~\cite{Oze_91}.

\subsection{Quintic self-dual codes over $\mathbb F_3$ and $\mathbb F_4$}
\label{subsec:quintic-Fs}

In this section, we find more quintic self-dual codes over $\mathbb F_3$ and $\mathbb F_4$
which are optimal or best known self-dual codes by using the building-up construction in
Theorem~\ref{thm:construction1}.

\begin{itemize}
\item  Case: $q=3$\\
Using (ii) of Theorem~\ref{thm:construction1} with $\alpha=1$ and
$\beta=1$, we obtain the following $I_8 = [L~|~R]$:

{\scriptsize{
\begin{equation*}
L = \left[
  \begin{array}{cccc}
1&   0&   0&   0\\
0&   1&   0&   0\\
Y^4 + 2Y^2 + Y&   2Y^4 + 2Y^3 + Y&   2Y^3 + 2Y^2 + 2Y&   2Y^4 + Y^3 + 2Y^2\\
Y^4 + 2Y^3 + Y^2 + 2Y + 1&   2Y^2 + Y&   Y^4 + 2Y^3 + 1&   Y^4 + 2Y^3 + 2Y^2 + Y + 1
  \end{array}
\right],
\end{equation*}
}}

{\scriptsize{
\begin{equation*}
R = \left[
  \begin{array}{cccc}
Y^2 + 1&   2Y^4 + Y^2 + Y + 2&   2Y^3&   Y^4 + 2Y^3 + 2Y + 1\\
2Y^4 + Y^2 + Y + 2&   2Y^4 + Y^3 + Y^2 + 2Y + 1&   Y^4 + 2Y^3 + 2Y^2 + 1&   Y^4 + 2Y^3 + Y + 1\\
 Y^4 + Y^2 + 2&   2Y^4 + Y^3 + 2Y^2 + 2Y + 1&   Y^4 + 2Y^3 + 2Y^2 + Y&   Y^3 + 2Y + 1\\
  2Y^4 + 2Y^3 + 2&   2Y^4 + Y^3 + 2Y^2 + 2&   Y&   2Y^3 + 2Y^2 + Y + 2
  \end{array}
\right].
\end{equation*}
}}

We also obtain a $2$ by $4$ matrix $I_4$ by deleting the first four
columns and the first two rows of $I_8$. The corresponding ternary
quasi-cyclic self-dual codes are all extremal self-dual codes. More
specifically, $I_4$ induces a $[20, 10, 6]$ code and $I_8$ induces a
$[40, 20, 12]$ code, and the orders of the automorphism groups are
$2^{8} \cdot 3 \cdot 5, ~10$, respectively. There are exactly six
extremal $[20, 10, 6]$ self-dual codes, and our code with the
generator matrix $I_4$ corresponds to $19$th code in Table
III~\cite{PleSloWar}.

We denote the code with the generator matrix $I_8$ by
$QSD_{40}^{3}$. There are at least 118 $[40, 20, 12]$ ternary
extremal self-dual codes. More precisely, the $15$ codes with
automorphisms of prime order $r > 5$ were found in~\cite{Huf2}. It
was reported in~\cite{Har} that there are five more $[40, 20, 12]$
ternary extremal self-dual codes. But we have checked that the codes
$C_{40,w1}$ and $C_{40,w3}$ in~\cite[Table 6]{Har} have minimum
weight $9$. Hence three codes were found in~\cite{Har}, and we have
verified that these three codes are not equivalent to
$QSD_{40}^{3}$. There are $100$ codes in~\cite{HHKK} whose
automorphism group orders are greater than
$|\mbox{Aut}(QSD_{40}^3)|=10$. In what follows, we give the
generator matrix $[L|R]$ of $QSD_{40}^{3}$: {\tiny{
\begin{equation*}
\label{}
L = \left[
  \begin{array}{cccccccccccccccccccc}
1& 0& 0& 0& 1& 2& 0& 1& 0& 0& 0& 0& 0& 1& 0& 2& 0& 0& 0& 0\\
0& 1& 0& 0& 2& 1& 1& 1& 0& 0& 0& 0& 1& 2& 0& 1& 0& 0& 0& 0\\
0& 0& 0& 0& 2& 1& 0& 1& 1& 1& 2& 0& 0& 2& 1& 2& 2& 0& 2& 2\\
1& 0& 1& 1& 2& 2& 0& 2& 2& 1& 0& 1& 0& 0& 1& 1& 1& 2& 0& 2\\
0& 0& 0& 0& 0& 2& 0& 1& 1& 0& 0& 0& 1& 2& 0& 1& 0& 0& 0& 0\\
0& 0& 0& 0& 2& 2& 1& 1& 0& 1& 0& 0& 2& 1& 1& 1& 0& 0& 0& 0\\
1& 2& 0& 2& 1& 2& 1& 0& 0& 0& 0& 0& 2& 1& 0& 1& 1& 1& 2& 0\\
1& 0& 1& 1& 2& 2& 0& 0& 1& 0& 1& 1& 2& 2& 0& 2& 2& 1& 0& 1\\
0& 0& 0& 0& 0& 0& 2& 2& 0& 0& 0& 0& 0& 2& 0& 1& 1& 0& 0& 0\\
0& 0& 0& 0& 0& 1& 2& 2& 0& 0& 0& 0& 2& 2& 1& 1& 0& 1& 0& 0\\
0& 2& 2& 1& 0& 1& 2& 1& 1& 2& 0& 2& 1& 2& 1& 0& 0& 0& 0& 0\\
2& 0& 2& 2& 2& 1& 0& 2& 1& 0& 1& 1& 2& 2& 0& 0& 1& 0& 1& 1\\
0& 0& 0& 0& 1& 1& 0& 0& 0& 0& 0& 0& 0& 0& 2& 2& 0& 0& 0& 0\\
0& 0& 0& 0& 1& 1& 2& 0& 0& 0& 0& 0& 0& 1& 2& 2& 0& 0& 0& 0\\
2& 0& 2& 2& 1& 2& 2& 0& 0& 2& 2& 1& 0& 1& 2& 1& 1& 2& 0& 2\\
1& 2& 0& 2& 0& 2& 0& 2& 2& 0& 2& 2& 2& 1& 0& 2& 1& 0& 1& 1\\
0& 0& 0& 0& 0& 1& 0& 2& 0& 0& 0& 0& 1& 1& 0& 0& 0& 0& 0& 0\\
0& 0& 0& 0& 1& 2& 0& 1& 0& 0& 0& 0& 1& 1& 2& 0& 0& 0& 0& 0\\
1& 1& 2& 0& 0& 2& 1& 2& 2& 0& 2& 2& 1& 2& 2& 0& 0& 2& 2& 1\\
2& 1& 0& 1& 0& 0& 1& 1& 1& 2& 0& 2& 0& 2& 0& 2& 2& 0& 2& 2
  \end{array}
\right],
\end{equation*}
}}
{\tiny{
\begin{equation*}
\label{}
R = \left[
  \begin{array}{cccccccccccccccccccc}
1& 1& 0& 0& 0& 0& 0& 0& 0& 0& 2& 2& 0& 0& 0& 0& 0& 2& 0& 1\\
1& 1& 2& 0& 0& 0& 0& 0& 0& 1& 2& 2& 0& 0& 0& 0& 2& 2& 1& 1\\
1& 2& 2& 0& 0& 2& 2& 1& 0& 1& 2& 1& 1& 2& 0& 2& 1& 2& 1& 0\\
0& 2& 0& 2& 2& 0& 2& 2& 2& 1& 0& 2& 1& 0& 1& 1& 2& 2& 0& 0\\
0& 1& 0& 2& 0& 0& 0& 0& 1& 1& 0& 0& 0& 0& 0& 0& 0& 0& 2& 2\\
1& 2& 0& 1& 0& 0& 0& 0& 1& 1& 2& 0& 0& 0& 0& 0& 0& 1& 2& 2\\
0& 2& 1& 2& 2& 0& 2& 2& 1& 2& 2& 0& 0& 2& 2& 1& 0& 1& 2& 1\\
0& 0& 1& 1& 1& 2& 0& 2& 0& 2& 0& 2& 2& 0& 2& 2& 2& 1& 0& 2\\
1& 2& 0& 1& 0& 0& 0& 0& 0& 1& 0& 2& 0& 0& 0& 0& 1& 1& 0& 0\\
2& 1& 1& 1& 0& 0& 0& 0& 1& 2& 0& 1& 0& 0& 0& 0& 1& 1& 2& 0\\
2& 1& 0& 1& 1& 1& 2& 0& 0& 2& 1& 2& 2& 0& 2& 2& 1& 2& 2& 0\\
2& 2& 0& 2& 2& 1& 0& 1& 0& 0& 1& 1& 1& 2& 0& 2& 0& 2& 0& 2\\
0& 2& 0& 1& 1& 0& 0& 0& 1& 2& 0& 1& 0& 0& 0& 0& 0& 1& 0& 2\\
2& 2& 1& 1& 0& 1& 0& 0& 2& 1& 1& 1& 0& 0& 0& 0& 1& 2& 0& 1\\
1& 2& 1& 0& 0& 0& 0& 0& 2& 1& 0& 1& 1& 1& 2& 0& 0& 2& 1& 2\\
2& 2& 0& 0& 1& 0& 1& 1& 2& 2& 0& 2& 2& 1& 0& 1& 0& 0& 1& 1\\
0& 0& 2& 2& 0& 0& 0& 0& 0& 2& 0& 1& 1& 0& 0& 0& 1& 2& 0& 1\\
0& 1& 2& 2& 0& 0& 0& 0& 2& 2& 1& 1& 0& 1& 0& 0& 2& 1& 1& 1\\
0& 1& 2& 1& 1& 2& 0& 2& 1& 2& 1& 0& 0& 0& 0& 0& 2& 1& 0& 1\\
2& 1& 0& 2& 1& 0& 1& 1& 2& 2& 0& 0& 1& 0& 1& 1& 2& 2& 0& 2
  \end{array}
\right].
\end{equation*}
}}

As a summary, we have the following theorem.
\begin{thm}
There are at least $119$ monomially inequivalent self-dual $[40, 20, 12]$ codes over $\F_3$.
\end{thm}


\item Case: $q=4$

Applying a similar process as before up to code length $\ell = 6$ with $c=1$,
we find the following $J_6 = [L~|~R]$:

{\scriptsize{
\begin{equation*}
L = \left[
  \begin{array}{ccc}
1&   0&   Y^4 + Y^3 + Y^2 + Y + 1 \\
\om Y^4 + \om ^2Y^3 + Y^2 + Y&   \om Y^4 + \om ^2Y^3 + Y^2 + Y& 1 \\
\om ^2Y^4 + \om ^2Y^3 + \om Y^2 + \om Y + \om &   \om ^2Y^4 + \om ^2Y^3 + \om Y^2 + \om Y + \om &  \om ^2Y^2 + Y
  \end{array}
\right],
\end{equation*}
}}

{\scriptsize{
\begin{equation*}
R = \left[
  \begin{array}{ccc}
 Y^4 + Y^3 + \om Y^2 + Y&   Y^4 + \om Y^3 + \om ^2Y^2 + \om ^2Y + 1&   \om ^2Y^4 + Y^3 + Y^2 + \om  \\
 0& Y^4 + Y^3 + Y^2 + Y + \om & \om ^2Y^4 + Y^3 + \om ^2Y^2 + \om Y \\
 \om ^2Y^2 + Y& Y^4 + Y^3 + Y^2 + Y + \om & Y^4 + \om Y^3 + \om ^2Y^2 + \om
  \end{array}
\right],
\end{equation*}
}}

where $\om$ is a generator of $\F_4^\ast$. The corresponding
quaternary quasi-cyclic Euclidean self-dual codes are all optimal or
have the best known parameters. See~\cite{Han_Data} for the generator
matrices of these quaternary codes. By successively deleting the first
two columns and the first row of $J_6$, we obtain $J_4$ and $J_2$.
More precisely, $J_2$ induces a $[10, 5, 4]$ code (optimal), $J_4$
induces a $[20, 10, 8]$ code (optimal), and $J_6$ induces a $[30,
15, 10]$ code (best known). The quaternary code corresponding to $J_4$ is equivalent to $XQ_{19}$~\cite{LinMac}. We denote the quaternary code corresponding to the generator matrix $J_6$ by $QSD_{30}^{4}$ whose generator matrix $G(QSE_{30}^4)$ is given below.
We have computed that $QSD_{30}^{4}$ has minimum distance $10$, $A_{10}=1893$, and the automorphism group of order $30$.
 As far as we know, only one self-dual $[30, 15, 10]$ code over $\F_4$ was known before, and that code is the one denoted by $(f_2; 11; 25)$~\cite{GO}. (It was reported to us that the code denoted by $(f_2; 11; 15)$~\cite{GO} is an error since it has minimum distance $6$.) The code $(f_2; 11; 25)$ has minimum distance $10$, $A_{10}=1854$, and the automorphism group of order $90$. Therefore the two codes $QSD_{30}^{4}$ and $(f_2; 11; 25)$ are not equivalent. We note that the minimum Lee weight $d_L$ of these codes in the sense of~\cite{BetGul} and \cite{GabPle} is $10$ and that only one self-dual $[30,15,9]$ code over $\F_4$ with $d_L=10$ is given in~\cite[Table VIII]{BetGul}.


As a summary, we have the following theorem.
\begin{thm}
There are at least two monomially inequivalent self-dual $[30, 15, 10]$ codes over $\F_4$.
\end{thm}

{\tiny{
\begin{equation*}
\label{}
G(QSD_{30}^{4})=
\left[
  \begin{array}{l}
1\ 0\ 1\ 0\ 1\ w\ 0\ 0\ 1\ 1\ w^2\ 0\ 0\ 0\ 1\ w\ w^2\ 1\ 0\ 0\ 1\ 1\ w\ 1\ 0\ 0\ 1\ 1\ 1\ w^2\\
0\ 0\ 1\ 0\ w\ 0\ 1\ 1\ 0\ 0\ 1\ w\ 1\ 1\ 0\ 0\ 1\ w^2\ w^2\ w^2\ 0\ 0\ 1\ 1\ w\ w\ 0\ 0\ 1\ w^2\\
w\ w\ 0\ 0\ w\ w\ w\ w\ 1\ 1\ 1\ 0\ w\ w\ w^2\ w^2\ 1\ w^2\ w^2\ w^2\ 0\ 0\ 1\ w\ w^2\ w^2\ 0\ 0\ 1\ 1\\
0\ 0\ 1\ 1\ 1\ w^2\ 1\ 0\ 1\ 0\ 1\ w\ 0\ 0\ 1\ 1\ w^2\ 0\ 0\ 0\ 1\ w\ w^2\ 1\ 0\ 0\ 1\ 1\ w\ 1\\
w\ w\ 0\ 0\ 1\ w^2\ 0\ 0\ 1\ 0\ w\ 0\ 1\ 1\ 0\ 0\ 1\ w\ 1\ 1\ 0\ 0\ 1\ w^2\ w^2\ w^2\ 0\ 0\ 1\ 1\\
w^2\ w^2\ 0\ 0\ 1\ 1\ w\ w\ 0\ 0\ w\ w\ w\ w\ 1\ 1\ 1\ 0\ w\ w\ w^2\ w^2\ 1\ w^2\ w^2\ w^2\ 0\ 0\ 1\ w\\\
0\ 0\ 1\ 1\ w\ 1\ 0\ 0\ 1\ 1\ 1\ w^2\ 1\ 0\ 1\ 0\ 1\ w\ 0\ 0\ 1\ 1\ w^2\ 0\ 0\ 0\ 1\ w\ w^2\ 1\\
w^2\ w^2\ 0\ 0\ 1\ 1\ w\ w\ 0\ 0\ 1\ w^2\ 0\ 0\ 1\ 0\ w\ 0\ 1\ 1\ 0\ 0\ 1\ w\ 1\ 1\ 0\ 0\ 1\ w^2\\
w^2\ w^2\ 0\ 0\ 1\ w\ w^2\ w^2\ 0\ 0\ 1\ 1\ w\ w\ 0\ 0\ w\ w\ w\ w\ 1\ 1\ 1\ 0\ w\ w\ w^2\ w^2\ 1\ w^2\\
0\ 0\ 1\ w\ w^2\ 1\ 0\ 0\ 1\ 1\ w\ 1\ 0\ 0\ 1\ 1\ 1\ w^2\ 1\ 0\ 1\ 0\ 1\ w\ 0\ 0\ 1\ 1\ w^2\ 0\\
1\ 1\ 0\ 0\ 1\ w^2\ w^2\ w^2\ 0\ 0\ 1\ 1\ w\ w\ 0\ 0\ 1\ w^2\ 0\ 0\ 1\ 0\ w\ 0\ 1\ 1\ 0\ 0\ 1\ w\\
w\ w\ w^2\ w^2\ 1\ w^2\ w^2\ w^2\ 0\ 0\ 1\ w\ w^2\ w^2\ 0\ 0\ 1\ 1\ w\ w\ 0\ 0\ w\ w\ w\ w\ 1\ 1\ 1\ 0\\
0\ 0\ 1\ 1\ w^2\ 0\ 0\ 0\ 1\ w\ w^2\ 1\ 0\ 0\ 1\ 1\ w\ 1\ 0\ 0\ 1\ 1\ 1\ w^2\ 1\ 0\ 1\ 0\ 1\ w\\
1\ 1\ 0\ 0\ 1\ w\ 1\ 1\ 0\ 0\ 1\ w^2\ w^2\ w^2\ 0\ 0\ 1\ 1\ w\ w\ 0\ 0\ 1\ w^2\ 0\ 0\ 1\ 0\ w\ 0\\
w\ w\ 1\ 1\ 1\ 0\ w\ w\ w^2\ w^2\ 1\ w^2\ w^2\ w^2\ 0\ 0\ 1\ w\ w^2\ w^2\ 0\ 0\ 1\ 1\ w\ w\ 0\ 0\ w\ w\\
 \end{array}
\right]
\end{equation*}
}}

\end{itemize}

\subsection{Septic self-dual codes over $\F_2, \F_4$, and $\F_5$}
\label{subsec:septic}

In this section, we find septic self-dual codes over $\mathbb F_q$
which are optimal or have the best known self-dual codes by using the building-up construction in
Theorem~\ref{thm:construction1}.

\begin{itemize}
\item  Case: $q=2$

We do a similar process as before up to the length $\ell = 8$ with
$c=1$, so we get $K_8 = [L | R]$ as follows: {\scriptsize{
\begin{equation*}
L = \left[
  \begin{array}{cccc}
1&   0&   Y^4 + Y^3 + Y^2&   Y^6 + Y^5 + Y^4 + Y^2 + Y + 1 \\
Y^6 + Y^5 + Y^3 + Y^2 + 1&   Y^6 + Y^5 + Y^3 + Y^2 + 1&   1&   0 \\
Y^5 + Y + 1&   Y^5 + Y + 1&   Y^6 + Y^4 + Y^3 + Y^2&   Y^6 + Y^4 + Y^3 + Y^2 \\
Y^5 + Y^4 + Y^3 + Y + 1&   Y^5 + Y^4 + Y^3 + Y + 1&   Y^6&   Y^6
\end{array}
  \right],
\end{equation*}
}}
{\scriptsize{
\begin{equation*}
R = \left[
  \begin{array}{cccc}
Y^4 + Y&  Y^6 + Y^4 + Y^3&   Y^6 + Y^3 + Y + 1&   Y^6 + Y^5 + Y^4 + Y^3 + 1 \\
Y^3 + 1&   Y^4 + Y^3 + Y^2 + 1&   Y^5 + Y^2 + Y&   Y^6 + Y^5 + Y^4 + 1 \\
1&   0&   Y^6 + Y^4 + Y + 1&   Y \\
Y^6 + Y^5 + Y^4 + Y^3 + Y^2&   Y^6 + Y^5 + Y^4 + Y^3 + Y^2&   Y^3 + Y^2 + 1&   Y^3 + Y + 1
\end{array}
  \right].
\end{equation*}
}} The corresponding binary quasi-cyclic self-dual codes are all
optimal self-dual codes. By successively deleting the first two
columns and the first row of $K_8$, we obtain $K_6, K_4,$ and $K_2$.
More specifically, $K_2$ induces a $[14, 7, 4]$ code, $K_4$ induces
a $[28, 14, 6]$ code, $K_6$ induces a $[42, 21, 8]$ code, and $K_8$
induces a Type II $[56, 28, 12]$ code. The weight enumerator of the
$[42, 21, 8]$ code corresponds to $W_2$ with $\beta =0$
in~\cite{Huf}.

\item Case: $q=4$

Doing a similar process as before up to the length $\ell = 6$ with
$c=1$, we find the following $M_6 = [L | R]$: {\tiny{
\begin{equation*}
L = \left[
  \begin{array}{ccc}
1&   0&   \om Y^6 + \om ^2Y^5 + Y^3 + Y + \om \\
\om ^2Y^5 + 1&   \om ^2Y^5 + 1&   1 \\
Y^6 + \om ^2Y^4 + \om ^2Y^2 + \om ^2Y + \om ^2&   Y^6 + \om ^2Y^4 + \om ^2Y^2 + \om ^2Y + \om ^2&   \om ^2Y^5 + Y^3 + Y^2 + Y + \om ^2
\end{array}
  \right],
\end{equation*}
}}
{\tiny{
\begin{equation*}
R = \left[
  \begin{array}{ccc}
\om ^2Y^6 + Y^5 + \om ^2Y^4 + Y^2 + \om Y + \om ^2&   \om ^2Y^6 + \om Y^5 + \om Y^2 + \om ^2Y&   Y^6 + \om Y^5 + Y^4 + \om ^2Y^3 + \om ^2Y + 1 \\
0&   \om Y^6 + Y^5 + \om ^2Y^3 + \om ^2Y^2&   Y^5 + \om Y^4 + \om Y^3 + \om ^2Y^2 \\
\om ^2Y^5 + Y^3 + Y^2 + Y + \om ^2&   Y^6 + \om Y^5 + \om Y^4 + Y^3 + \om ^2Y^2 + Y + \om ^2&   \om ^2Y^6 + Y^5 + \om ^2Y^3 + \om ^2Y^2 + \om ^2
\end{array}
  \right].
\end{equation*}
}}

The corresponding quaternary quasi-cyclic self-dual codes are all
optimal or have the best known parameters. By successively deleting
the first two columns and the first row of $M_6$, we obtain $M_4$
and $M_2$. More specifically, $M_2$ induces an optimal self-dual
$[14, 7, 6]$ code over $\F_4$, $M_4$ induces a self-dual code over
$\F_4$ with the best known parameters $[28, 14, 9]$, and $M_6$
induces a self-dual code over $\F_4$ with the best known parameters
$[42, 21, 12]$. We denote these codes by $SSD_{14}^4, SSD_{28}^4,
SSD_{42}^4$, respectively. We verified that $SSD_{14}^4$ is
equivalent to $QDC_{14}$~\cite{Gab} which is the only known
self-dual $[14,7,6]$ code over $\F_4$.\vspace{5pt}

Only two self-dual $[28, 14, 9]$ codes over $\F_4$ were known, and one is $XQ_{27}$~\cite{LinMac} and the other is $D_{II,28}$~\cite{BetGul}.
The number $A_9$ of minimum weight codewords of $XQ_{27}$ ($D_{II,28}$, respectively) is $3276$ ($1092$, respectively). On the other hand, our code $SSD_{28}^4$ has $A_9=630$. This shows that $SSD_{28}^4$ is a new code. Furthermore, we have checked that $SSD_{28}^4$ is a Type II code over $\F_4$ with minimum Lee weight $d_L=12$. We recall that a Euclidean self-dual code over $\F_4$ is called {\em Type II} if its binary image under the Gray map $\phi$ is Type II (see~\cite{GabPle}), where
the Gray map $\phi$ from $GF(4)^n$ to $GF(2)^{2n}$ is defined as $\phi(\om{\bf{x}} + \ob{\bf{y}}) = ({\bf{x}}, {\bf{y}})$ for ${\bf{x}}, {\bf{y}} \in GF(2)^n$ and $({\bf{x}}, {\bf{y}})$ is the binary vector of length $2n$. We have calculated that
$|{\mbox{Aut}}(\phi(SSD_{28}^4))|=7$,
$|{\mbox{Aut}}(\phi(D_{II,28}))|=28$, and
$|{\mbox{Aut}}(\phi(XQ_{27}))|=2^3 \cdot 3^4 \cdot 7 \cdot 13$.\vspace{5pt}

We have also checked that both $D_{II,28}$ and $XQ_{27}$ are Type II codes over $\F_4$ with $d_L=12$.
We therefore find that there are at least three Lee-extremal Type II $[28,14, d_L=12]$ codes over $\F_4$.\vspace{5pt}

We are aware of two papers~\cite{BouUlm} and \cite{CCN}, in which six Euclidean self-dual $[28,14,9]$ codes over $\mathbb F_4$ are known to exist.  However their generator matrices and the number of minimum weight codewords are not given explicitly. Hence we omit the equivalence check of their codes with $SSD_{28}^4$.

For length $42$, there has been only one self-dual $[42, 21, 12]$
code over $\F_4$, denoted by $(f_2; 11; 17)$~\cite{GO}. This code
has $A_{12}=945$, but our code $SSD_{42}^4$ has $A_{12}=323$ and
$d_L=12$. Hence they are inequivalent, and this implies that
$SSD_{42}^4$ is a new code.\vspace{5pt}

In what follows, we give the generator matrix $[L|R]$ of $SSD_{28}^4$:

{\tiny{
\begin{equation*}
\label{}
L = \left[
  \begin{array}{cccccccccccccc}
1& 0& 0& 0& 0& 0& 0& 0& 0& 0& \w^2& \w^2& 0& 0\\
\w^2& \w^2& \w^2& \w^2& 1& 1& 1& 0& 1& 1& \w^2& \w^2& 1& 1\\
0& 0& \w& 0& 1& 0& 0& 0& 0& 0& 0& 0& 0& 0\\
0& 0& 1& \w^2& \w^2& \w^2& \w^2& \w^2& 1& 1& 1& 0& 1& 1\\
0& 0& 1& 1& 0& 0& \w& 0& 1& 0& 0& 0& 0& 0\\
\w^2& \w^2& \w& 1& 0& 0& 1& \w^2& \w^2& \w^2& \w^2& \w^2& 1& 1\\
0& 0& 0& \w& 0& 0& 1& 1& 0& 0& \w& 0& 1& 0\\
0& 0& \w& 0& \w^2& \w^2& \w& 1& 0& 0& 1& \w^2& \w^2& \w^2\\
0& 0& \w^2& \w& 0& 0& 0& \w& 0& 0& 1& 1& 0& 0\\
1& 1& 1& \w^2& 0& 0& \w& 0& \w^2& \w^2& \w& 1& 0& 0\\
0& 0& \w^2& \w^2& 0& 0& \w^2& \w& 0& 0& 0& \w& 0& 0\\
1& 1& \w^2& \w^2& 1& 1& 1& \w^2& 0& 0& \w& 0& \w^2& \w^2\\
0& 0& 0& 0& 0& 0& \w^2& \w^2& 0& 0& \w^2& \w& 0& 0\\
1& 1& 1& 0& 1& 1& \w^2& \w^2& 1& 1& 1& \w^2& 0& 0
  \end{array}
\right],
\end{equation*}
}}

{\tiny{
\begin{equation*}
\label{}
R = \left[
  \begin{array}{cccccccccccccc}
\w^2& \w& 0& 0& 0& \w& 0& 0& 1& 1& 0& 0& \w& 0\\
1& \w^2& 0& 0& \w& 0& \w^2& \w^2& \w& 1& 0& 0& 1& \w^2\\
\w^2& \w^2& 0& 0& \w^2& \w& 0& 0& 0& \w& 0& 0& 1& 1\\
\w^2& \w^2& 1& 1& 1& \w^2& 0& 0& \w& 0& \w^2& \w^2& \w& 1\\
0& 0& 0& 0& \w^2& \w^2& 0& 0& \w^2& \w& 0& 0& 0& \w\\
1& 0& 1& 1& \w^2& \w^2& 1& 1& 1& \w^2& 0& 0& \w& 0\\
0& 0& 0& 0& 0& 0& 0& 0& \w^2& \w^2& 0& 0& \w^2& \w\\
\w^2& \w^2& 1& 1& 1& 0& 1& 1& \w^2& \w^2& 1& 1& 1& \w^2\\
\w& 0& 1& 0& 0& 0& 0& 0& 0& 0& 0& 0& \w^2& \w^2\\
1& \w^2& \w^2& \w^2& \w^2& \w^2& 1& 1& 1& 0& 1& 1& \w^2& \w^2\\
1& 1& 0& 0& \w& 0& 1& 0& 0& 0& 0& 0& 0& 0\\
\w& 1& 0& 0& 1& \w^2& \w^2& \w^2& \w^2& \w^2& 1& 1& 1& 0\\
0& \w& 0& 0& 1& 1& 0& 0& \w& 0& 1& 0& 0& 0\\
\w& 0& \w^2& \w^2& \w& 1& 0& 0& 1& \w^2& \w^2& \w^2& \w^2& \w^2
  \end{array}
\right].
\end{equation*}
}}

We also give the generator matrix $[L|R]$ of $SSD_{42}^4$ in the following:

{\tiny{
\begin{equation*}
\label{}
L = \left[
  \begin{array}{ccccccccccccccccccccc}
1& 0& \w& \w^2& 0& 1& 0& 0& 1& \w& \w^2& \w^2& 0& 0& 0& 1& \w& 0& 0& 0& 1\\
1& 1& 1& 0& 0& 0& 0& 0& 0& 0& 0& 0& 0& 0& 0& 0& \w^2& \w^2& 0& 0& 0\\
\w^2& \w^2& \w^2& \w^2& \w^2& \w^2& \w^2& \w^2& 1& 1& 1& 0& \w^2& \w^2& 1& 1& \w^2& \w^2& 0& 0& 1\\
0& 0& \w& \w^2& \w^2& 1& 1& 0& \w& \w^2& 0& 1& 0& 0& 1& \w& \w^2& \w^2& 0& 0& 0\\
0& 0& 0& 0& \w& 0& 1& 1& 1& 0& 0& 0& 0& 0& 0& 0& 0& 0& 0& 0& 0\\
1& 1& 0& 0& 1& \w^2& \w^2& \w^2& \w^2& \w^2& \w^2& \w^2& \w^2& \w^2& 1& 1& 1& 0& \w^2& \w^2& 1\\
0& 0& \w^2& 1& \w& \w& 0& 0& \w& \w^2& \w^2& 1& 1& 0& \w& \w^2& 0& 1& 0& 0& 1\\
\w^2& \w^2& 0& 0& 1& 1& 0& 0& 0& 0& \w& 0& 1& 1& 1& 0& 0& 0& 0& 0 &0 \\
0& 0& \w^2& \w^2& \w& 1& 1& 1& 0& 0& 1& \w^2& \w^2& \w^2& \w^2& \w^2& \w^2& \w^2& \w^2& \w^2& 1\\
0& 0& 0& \w^2& 0& 1& 0& 0& \w^2& 1& \w& \w& 0& 0& \w& \w^2& \w^2& 1& 1& 0& \w\\
0& 0& 0& 0& 0& \w& \w^2& \w^2& 0& 0& 1& 1& 0& 0& 0& 0& \w& 0& 1& 1& 1\\
\w^2& \w^2& 0& 0& \w& 0& 0& 0& \w^2& \w^2& \w& 1& 1& 1& 0& 0& 1& \w^2& \w^2& \w^2& \w^2\\
0& 0& 1& 0& 0& \w^2& 0& 0& 0& \w^2& 0& 1& 0& 0& \w^2& 1& \w& \w& 0& 0& \w\\
0& 0& 0& 0& \w^2& \w& 0& 0& 0& 0& 0& \w& \w^2& \w^2& 0& 0& 1& 1& 0& 0& 0\\
0& 0& 1& 1& 1& \w^2& \w^2& \w^2& 0& 0& \w& 0& 0& 0& \w^2& \w^2& \w& 1& 1& 1& 0\\
0& 0& 0& 1& \w& 0& 0& 0& 1& 0& 0& \w^2& 0& 0& 0& \w^2& 0& 1& 0& 0& \w^2\\
0& 0& 0& 0& \w^2& \w^2& 0& 0& 0& 0& \w^2& \w& 0& 0& 0& 0& 0& \w& \w^2& \w^2& 0\\
\w^2& \w^2& 1& 1& \w^2& \w^2& 0& 0& 1& 1& 1& \w^2& \w^2& \w^2& 0& 0& \w& 0& 0& 0& \w^2\\
0& 0& 1& \w& \w^2& \w^2& 0& 0& 0& 1& \w& 0& 0& 0& 1& 0& 0& \w^2& 0& 0& 0\\
0& 0& 0& 0& 0& 0& 0& 0& 0& 0& \w^2& \w^2& 0& 0& 0& 0& \w^2& \w& 0& 0& 0\\
\w^2& \w^2& 1& 1& 1& 0& \w^2& \w^2& 1& 1& \w^2& \w^2& 0& 0& 1& 1& 1& \w^2& \w^2& \w^2& 0
  \end{array}
\right],
\end{equation*}
}}

{\tiny{
\begin{equation*}
\label{}
R = \left[
  \begin{array}{ccccccccccccccccccccc}
0& 0& \w^2& 0& 0& 0& \w^2& 0& 1& 0& 0& \w^2& 1& \w& \w& 0& 0& \w& \w^2& \w^2& 1\\
0& \w^2& \w& 0& 0& 0& 0& 0& \w& \w^2& \w^2& 0& 0& 1& 1& 0& 0& 0& 0& \w& 0\\
1& 1& \w^2& \w^2& \w^2& 0& 0& \w& 0& 0& 0& \w^2& \w^2& \w& 1& 1& 1& 0& 0& 1& \w^2\\
1& \w& 0& 0& 0& 1& 0& 0& \w^2& 0& 0& 0& \w^2& 0& 1& 0& 0& \w^2& 1& \w& \w\\
0& \w^2& \w^2& 0& 0& 0& 0& \w^2& \w& 0& 0& 0& 0& 0& \w& \w^2& \w^2& 0& 0& 1& 1\\
1& \w^2& \w^2& 0& 0& 1& 1& 1& \w^2& \w^2& \w^2& 0& 0& \w& 0& 0& 0& \w^2& \w^2& \w& 1\\
\w& \w^2& \w^2& 0& 0& 0& 1& \w& 0& 0& 0& 1& 0& 0& \w^2& 0& 0& 0& \w^2& 0& 1\\
0& 0& 0& 0& 0& 0& 0& \w^2& \w^2& 0& 0& 0& 0& \w^2& \w& 0& 0& 0& 0& 0& \w\\
1& 1& 0& \w^2& \w^2& 1& 1& \w^2& \w^2& 0& 0& 1& 1& 1& \w^2& \w^2& \w^2& 0& 0& \w& 0\\
\w^2& 0& 1& 0& 0& 1& \w& \w^2& \w^2& 0& 0& 0& 1& \w& 0& 0& 0& 1& 0& 0& \w^2\\
0& 0& 0& 0& 0& 0& 0& 0& 0& 0& 0& 0& 0& \w^2& \w^2& 0& 0& 0& 0& \w^2& \w\\
\w^2& \w^2& \w^2& \w^2& \w^2& 1& 1& 1& 0& \w^2& \w^2& 1& 1& \w^2& \w^2& 0& 0& 1& 1& 1& \w^2\\
\w^2& \w^2& 1& 1& 0& \w& \w^2& 0& 1& 0& 0& 1& \w& \w^2& \w^2& 0& 0& 0& 1& \w& 0\\
0& \w& 0& 1& 1& 1& 0& 0& 0& 0& 0& 0& 0& 0& 0& 0& 0& 0& 0& \w^2& \w^2\\
0& 1& \w^2& \w^2& \w^2& \w^2& \w^2& \w^2& \w^2& \w^2& \w^2& 1& 1& 1& 0& \w^2& \w^2& 1& 1& \w^2& \w^2\\
1& \w& \w& 0& 0& \w& \w^2& \w^2& 1& 1& 0& \w& \w^2& 0& 1& 0& 0& 1& \w& \w^2& \w^2\\
0& 1& 1& 0& 0& 0& 0& \w& 0& 1& 1& 1& 0& 0& 0& 0& 0& 0& 0& 0& 0\\
\w^2& \w& 1& 1& 1& 0& 0& 1& \w^2& \w^2& \w^2& \w^2& \w^2& \w^2& \w^2& \w^2& \w^2& 1& 1& 1& 0\\
\w^2& 0& 1& 0& 0& \w^2& 1& \w& \w& 0& 0& \w& \w^2& \w^2& 1& 1& 0& \w& \w^2& 0& 1\\
0& 0& \w& \w^2& \w^2& 0& 0& 1& 1& 0& 0& 0& 0& \w& 0& 1& 1& 1& 0& 0& 0\\
0& \w& 0& 0& 0& \w^2& \w^2& \w& 1& 1& 1& 0& 0& 1& \w^2& \w^2& \w^2& \w^2& \w^2& \w^2& \w^2
  \end{array}
\right].
\end{equation*}
}}

The above results are summarized as follows.
\begin{thm}
There are at least three monomially inequivalent Euclidean self-dual $[28, 14, 9]$ codes over $\F_4$,
all of which are Lee-extremal Type II. There are at least two monomially inequivalent Euclidean self-dual $[42, 21, 12]$ codes over $\F_4$.
\end{thm}

\item Case: $q=5$

Doing a similar process as before up to the length $\ell = 4$ with $c=2$,
we obtain the following $N_4 = [L | R]$:

{\scriptsize{
\begin{equation*}
L = \left[
  \begin{array}{cc}
1&   0 \\
2Y^5 + 4Y^4 + Y^3 + Y + 1&   4Y^5 + 3Y^4 + 2Y^3 + 2Y + 2
\end{array}
  \right],
\end{equation*}
}}

{\scriptsize{
\begin{equation*}
R = \left[
  \begin{array}{cc}
3Y^5 + 2Y^4 + Y^3 + 3Y^2 + 4Y&   4Y^6 + 3Y^4 + 3Y^3 + Y^2 + 3Y + 1 \\
Y^4 + 3Y^3 + Y^2 + 4Y + 3&   Y^6 + 2Y^5 + 4Y^4 + 4Y^3 + 3Y^2 + 2Y + 3
\end{array}
  \right].
\end{equation*}
}}

The corresponding quaternary quasi-cyclic self-dual codes are all
optimal or have best known parameters. By deleting the first two
columns and the first row of $N_4$, we obtain $N_2$. More
specifically, $N_2$ induces an optimal self-dual $[14, 7, 6]$ code
over $\F_5$, and $N_4$ induces a self-dual code over $\F_5$ with the
best known parameters $[28, 14, 10]$ code, denoted by $SSD_{28}^5$.
We checked that $SSD_{28}^5$ is monomially equivalent to $Q_{28,4}$
in~\cite{GulHarMiy}.
\end{itemize}

\section*{Acknowledgment}
We thank an anonymous referee for his/her helpful
comments, which improved the clarity of this paper.
We also give thanks to M. Harada and A. Munemasa for a
discussion of the subsection~\ref{susec:cubic-various} in this
paper, to C. Huffman for helpful comments, and
to Boran Kim for her help in some part of computations done in
this paper.

{\comment{
Sunghyu Han \\
School of Liberal Arts \\
Korea University of Technology and Education \\
Cheonan 330-708, South Korea \\
{Email: \tt sunghyu@kut.ac.kr}\\

Jon-Lark Kim \\
Department of Mathematics  \\
University of Louisville \\
Louisville, KY 40292, USA\\
{Email: \tt jl.kim@louisville.edu} \\

Heisook Lee \\
Department of Mathematics  \\
Ewha Womans University \\
Seoul 120-750, South Korea\\
{Email: \tt hsllee@ewha.ac.kr}\\

Yoonjin Lee \\
Department of Mathematics  \\
Ewha Womans University \\
Seoul 120-750, South Korea\\
{Email: \tt yoonjinl@ewha.ac.kr}\\
}}

\end{document}